\documentclass{article}
\usepackage{graphicx}
\usepackage{amsmath}
\usepackage{amssymb}
\usepackage{amsthm}
\usepackage[colorlinks, citecolor={blue!0!black}]{hyperref}
\usepackage{pifont}
\usepackage[authoryear]{natbib}
\usepackage{xcolor}
\usepackage{amsmath, booktabs, graphicx, subcaption, array, tabularx, blkarray}
\usepackage[capitalize]{cleveref}
\usepackage{geometry}
\usepackage{algorithm}
\usepackage{algorithmic}
\usepackage{placeins}
\usepackage{svg}
\usepackage{tikz}
\usetikzlibrary{matrix}
\usetikzlibrary{calc}
\usetikzlibrary{positioning}

\newtheorem{theorem}{Theorem}[section]

\newtheorem{definition}[theorem]{Definition}

\newcommand{\cmark}{\ding{51}}%
\newcommand{\xmark}{\ding{55}}%

\begin{document}

\title{SVD Incidence Centrality: A Unified Spectral Framework for Node and Edge Analysis in Directed Networks and Hypergraphs}

\author{Jorge Luiz Franco$^{1,2}$ \and Thomas Peron$^{2}$ \and Alcebiades Dal Col$^{3}$ \and Fabiano Petronetto$^{3}$  \and Filipe Alves Neto Verri$^{1}$ \and Eric K. Tokuda$^{2}$ \and Luiz Gustavo Nonato$^{2}$}

\date{}

\maketitle

\begin{center}
\small
$^1$ Instituto Curvelo\\
$^2$ University of São Paulo\\
$^3$Federal University of Espirito Santo\\
\end{center}

\begin{abstract}
Identifying influential nodes and edges in directed networks remains a fundamental challenge across domains from social influence to biological regulation. Most existing centrality measures face a critical limitation: they either discard directional information through symmetrization or produce sparse, implementation-dependent rankings that obscure structural importance. We introduce a unified spectral framework for centrality analysis in directed networks grounded in the Singular value decomposition of the incidence matrix. The proposed approach derives both vertex and edge centralities via the pseudoinverse of Hodge Laplacians, yielding dense and well-resolved rankings that overcome the sparsity limitations commonly observed in betweenness centrality for directed graphs. Unlike traditional measures that require graph symmetrization, our framework naturally preserves directional information, enabling principled hub/authority analysis while maintaining mathematical consistency through spectral graph theory. The method extends naturally to hypergraphs through the same mathematical foundation. Experimental validation on real-world networks demonstrates the framework's effectiveness across diverse domains where traditional centrality measures encounter limitations due to implementation dependencies and sparse outputs.
\end{abstract}

\section{Introduction}

A fundamental challenge in network science is developing centrality measures that effectively handle both edge importance and directional relationships in complex networks. While traditional approaches excel for undirected graphs, they face limitations with directional information preservation. Current-flow centrality, eigenvector centrality, communicability centrality, and related measures require graph symmetrization for directed networks, potentially losing the asymmetric relationships that define social influence, biological regulation, and information flow systems \citep{kleinberg1999authoritative, freeman1978centrality}. Recent work on spectral centralities for directed and higher-order networks further highlights that symmetrization can obscure essential structural information encoded in orientation and flow \citep{contreras2024beyond}.

This limitation reveals a theoretical gap. Existing centrality frameworks treat vertex and edge importance as separate problems, requiring distinct computational approaches and often yielding inconsistent results. The lack of mathematical unity between vertex and edge measures can hinder coherent structural interpretation in systems where directionality, hierarchy, and flow are fundamental, such as regulatory, communication, and infrastructure networks \citep{bonacich1987power, newman2004finding}. Moreover, edge-centric measures have proved crucial in predictive tasks \citep{franco2026holding}.

We propose a unified spectral framework based on Hodge theory. We ground network centrality in the Singular Value Decomposition (SVD) of the incidence matrix, where directional information is naturally preserved through oriented boundary operators. This approach simultaneously derives vertex and edge centralities from the same mathematical foundation, ensuring consistency while maintaining theoretical rigor through connections to electrical network theory and algebraic topology. Related Hodge-theoretic formulations have recently been shown to provide interpretable decompositions of network flows into gradient, cyclic, and harmonic components, emphasizing the importance of orientation-preserving operators for network analysis \citep{schaub2020random}.

For undirected networks, the proposed SVD vertex centrality rankings equal current-flow closeness rankings—validating the electrical network foundations. However, this equivalence breaks down for directed networks, precisely where traditional measures face challenges. SVD centrality preserves directional information through spectral decomposition, enabling natural hub/authority analysis without graph symmetrization and providing a principled alternative to existing directed spectral constructions \citep{fanuel2019deformed}. We also provide a natural and direct extension of the framework to hypergraphs.

Our contributions span three key dimensions. Theoretically, we establish rigorous connections between centrality measures and physical observables through electrical resistance analogies and energy landscape interpretations derived from Hodge theory. Methodologically, we develop a unified spectral framework that generates mathematically consistent vertex and edge centralities from a single incidence matrix decomposition. Empirically, we validate the framework across diverse directed network domains, including social, biological, and infrastructure systems, demonstrating how orientation-preserving spectral centrality reveals structural information hidden by symmetrized approaches. 

This work explores spectral centrality as an approach for directed networks and hypergraphs analysis, connecting mathematical foundations with practical network science applications.

\section{Related Work}

Centrality measures have long been a key area of research in network science, quantifying the importance of vertices, edges, and hyperedges. Traditional vertex-based approaches like \emph{degree centrality} \citep{freeman1978centrality} evaluate the importance of a vertex based on its number of connections, a simple yet effective measure. However, such local approaches are insufficient to capture the broader influence of a vertex in the network. More sophisticated methods, such as \emph{eigenvector centrality} \citep{bonacich1987power} and \emph{PageRank} \citep{brin1998pagerank}, leverage the network's global structure. In particular, eigenvector centrality assigns higher scores to vertices connected to other influential vertices, capturing a more nuanced picture of importance. Similarly, \emph{PageRank} modifies this idea by incorporating directed edges, making it particularly effective for ranking web pages. However, PageRank primarily focuses on vertex centrality and does not directly address edge importance or hypergraph structures \citep{kucharczuk2022pagerank}.

Edge centrality measures like \emph{edge betweenness} \citep{newman2004finding, Lu2013edgebetw} emphasize the role of edges that lie on many shortest paths between vertices, highlighting connections in the network and also have hypergraph extensions \citep{hyperbetw}. However, betweenness centrality has limitations—it is based on shortest paths and is unable to capture global, spectral properties of the network, which can be important in more complex structures such as hypergraphs. Additionally, because betweenness centrality is based on shortest-path assumptions, its interpretability can be compromised when the real network process violates these assumptions \citep{bockholt2021systematic}

Recent extensions have addressed the complexities of hypergraphs, where edges can connect multiple vertices at once. For example, \citet{hyperbetw} extends betweenness centrality to hypergraphs, yet it still suffers from the same path-based limitations. On the other hand, \citet{eigenHypergraph} applies spectral methods to hypergraphs, leveraging eigenvector centrality to capture higher-order interactions. However, this method focuses solely on vertex centrality, without directly measuring edge importance.

A more advanced approach is presented by \citet{tudisco2021hyperedge}, which introduces a \emph{nonlinear eigenvector centrality} for both vertices and hyperedges. This method uses a spectral perspective to assess the centrality of vertices and edges simultaneously, providing a more global view of network structure. Nonetheless, the approach does not explicitly account for directed edges, limiting its application in certain scenarios.

In directed graphs, \citet{brohl2022straightforward} propose a neural network-based centrality measure (NN edge), focusing exclusively on edge centrality. While effective in identifying key edges, this approach lacks the ability to incorporate spectral properties, potentially missing the broader significance of edges in the network's global structure.

Spectral methods have also been extended to hypergraphs. \citet{eigenHypergraph} creates three eigenvector centralities for hypergraphs using tensors as representation. \citet{contreras2024beyond} introduces the \emph{HEC} and \emph{ZEC} methods, which use the Perron-Frobenius theorem and tensor-based algebra to compute centralities in hypergraphs. However, the tensor-based algebra adds complexity to the computation, making these methods less accessible and harder to interpret for large-scale networks.

\begin{table}
\centering
\caption{Taxonomy of popular centrality measures in directed graphs and hypergraphs.}
\label{tab:centrality_taxonomy}
\small
\begin{tabular*}{\textwidth}{@{\extracolsep{\fill}}lcccccc@{}}
\toprule
\textbf{Centrality Measure} & \textbf{Vertex} & \textbf{Edge} & \textbf{Hypergraph} & \textbf{Direction} & \textbf{Spectral} \\
\midrule
\parbox[t]{0.2\linewidth}{PageRank \\ \citep{brin1998pagerank}}         & \cmark & \xmark & \xmark & \cmark & \cmark \\
\midrule
\parbox[t]{0.2\linewidth}{Betweenness \\ \citep{hyperbetw}} & \cmark & \cmark & \cmark & \cmark & \xmark \\
\midrule
\parbox[t]{0.2\linewidth}{EigenHypergraph\\
\citep{eigenHypergraph}} & \cmark & \xmark & \cmark & \cmark & \cmark \\
\midrule
\parbox[t]{0.2\linewidth}{GRC \\ \citep{dal2023graph}} & \cmark & \xmark & \cmark & \xmark & \cmark \\
\midrule
\parbox[t]{0.2\linewidth}{NN edge \\
\citep{brohl2022straightforward}} & \xmark & \cmark & \xmark & \xmark & \xmark \\
\midrule
\parbox[t]{0.2\linewidth}{Nonlinear eigen \\
\citep{tudisco2021hyperedge}} & \cmark & \cmark & \cmark & \xmark & \cmark \\
\midrule
\parbox[t]{0.2\linewidth}{HEC/ZEC \\
\citep{contreras2024beyond}} & \cmark & \xmark & \cmark & \cmark & \cmark \\
\midrule
\parbox[t]{0.2\linewidth}{Matrix \\
\citep{vasilyeva2024matrix}} & \cmark & \xmark & \cmark & \cmark & \xmark \\
\midrule
\parbox[t]{0.2\linewidth}{Simplicial DualRank \\
\citep{liu2023eigenvector}} & \cmark & \xmark & \cmark & \cmark & \xmark \\
\midrule
\parbox[t]{0.2\linewidth}{SVD Incidence (Ours)} & \cmark & \cmark & \cmark & \cmark & \cmark \\
\bottomrule
\end{tabular*}
\end{table}

\emph{Graph Regularization Centrality (GRC)} proposed by~\citet{dal2023graph} takes a different route by focusing on vertex centrality in graphs using graph regularization and Fourier transform. However, it does not account for directionality, limiting its application to undirected networks.

\citet{vasilyeva2024matrix} presents a matrix-based approach to centrality in hypergraphs, providing a computationally efficient method that applies to both vertices and hyperedges. However, this method lacks a spectral component, making it less suitable for applications where global properties of networks are important.

A recent contribution by \citet{liu2023eigenvector} introduces the \emph{Simplicial DualRank centrality}, a parameter-free eigenvector centrality measure for weighted hypergraphs. Although the method incorporates vertex (inner) and hyperedge (outer) centrality, it does not account for direction.

These developments reveal a gap in approaches that integrate both vertex and edge centrality, particularly in directed graphs and hypergraphs. This gap motivates our proposed \emph{SVD Incidence Centrality} framework. Our framework leverages the singular value decomposition (SVD) of the incidence matrix to assess both vertex and edge centrality simultaneously. It captures the global spectral properties of both vertices and edges, providing a unified framework for directed graphs and hypergraphs.

\section{Preliminaries and Notation}


In this section, we explain the construction of the Hodge Laplacian for simplicial complexes, with a focus on how the incidence matrix serves as the boundary operator in this context. We start by establishing notations and definitions used throughout this paper, summarized in Table~\ref{tab:notation}.

\begin{table}[ht]
\centering
\caption{Table of Symbols and Notation.}
\label{tab:notation}
\begin{tabular}{ll}
\toprule
\textbf{Symbol} & \textbf{Description} \\
\midrule
$G=(V,E)$ & A graph with vertex set $V$ and edge set $E$. \\
$n = |V|$ & Number of vertices. \\
$m = |E|$ & Number of edges. \\
$B \in \mathbb{R}^{n \times m}$ & The oriented incidence matrix of the graph. \\
$L_0 = BB^\top$ & The 0-th Hodge Laplacian (Graph Laplacian on vertices). \\
$L_1 = B^\top B$ & The 1st Hodge Laplacian (Graph Laplacian on edges). \\
$B = U \Sigma V^\top$ & Singular Value Decomposition (SVD) of the incidence matrix. \\
$\sigma_k$ & The $k$-th singular value. \\
$u_k, v_k$ & The $k$-th left and right singular vectors, respectively. \\
$L_0^+, L_1^+$ & Moore-Penrose pseudoinverse of $L_0$ and $L_1$. \\
$C_v(i)$ & Vertex centrality of node $i$. \\
$C_e(e)$ & Edge centrality of edge $e$. \\
$R_{ij}$ & Effective resistance between nodes $i$ and $j$. \\
\bottomrule
\end{tabular}
\end{table}

For a directed graph with $n$ vertices and $m$ edges, the incidence matrix $B \in \mathbb{R}^{n \times m}$ has entries:
\begin{equation} \label{incidence_eq}
B_{ij} =
\begin{cases}
1 & \text{if edge } j \; \text{enters vertex } i \\
-1 & \text{if edge } j \; \text{leaves vertex } i \\
0 & \text{otherwise.}
\end{cases}
\end{equation}

The incidence matrix encodes not only undirected graphs, but also directed graphs and (un)directed hypergraphs. The orientation provided by the direction in directed graphs satisfies the fundamental boundary axiom in algebraic topology, where the composition of consecutive boundary operators equals zero: $\partial_1 \circ \partial_0 = 0$. This property is preserved in our framework because when we do not consider higher-order simplices (triangles), the boundary operator from edges to vertices ($\partial_1$) naturally satisfies this axiom through the incidence matrix structure. The directional information encoded in the ±1 entries of the incidence matrix provides the necessary orientation~\footnote{Here, the orientation of a graph is a choice of direction for each edge and it does not affect the underlying graph structure.} that makes this boundary operator well-defined, establishing a rigorous connection between graph orientation and simplicial complex theory.

\cref{fig:directed_incidence} displays the corresponding incidence matrix of a directed graph.

\begin{figure}[ht]
\centering
\includegraphics[width=.5\textwidth]{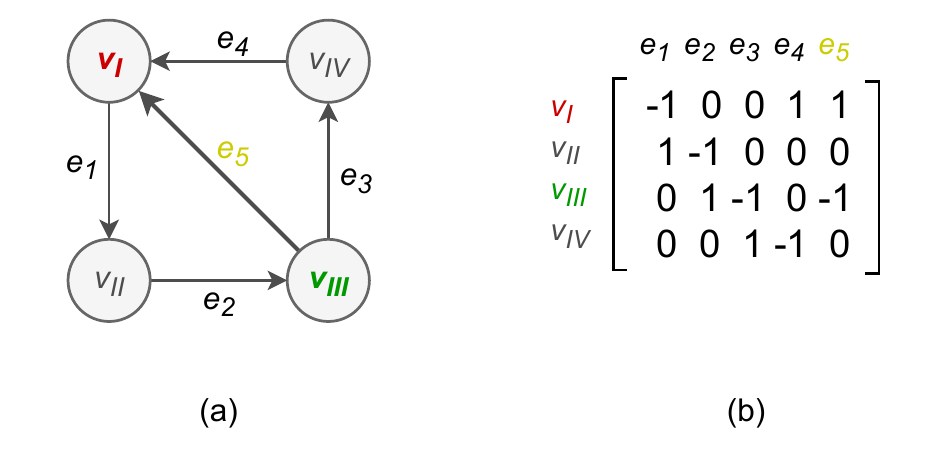}
\caption{A directed graph and its corresponding incidence matrix. For edge $e_5$ (highlighted in yellow), which leaves vertex $v_\text{III}$ and enters vertex $v_\text{I}$, the corresponding column in the matrix has a $-1$ at the $v_\text{III}$ (green) row and a $+1$ at the $v_\text{I}$ (red) row.}
\label{fig:directed_incidence}
\end{figure}

\subsection{Simplicial Complexes and Cochains}

A $k$-simplex is a set of $k+1$ vertices, denoted by $\{ v_0, v_1, \ldots, v_k \}$, which generalizes the concept of points, edges, triangles, and higher-dimensional counterparts. For instance:

\begin{itemize}
    \item A 0-simplex is a point (a vertex).
    \item A 1-simplex is an edge connecting two vertices.
    \item A 2-simplex is a triangle formed by three vertices.
    \item A 3-simplex is a tetrahedron formed by four vertices.
\end{itemize}

A face of a $k$-simplex is any subset of its vertices that forms a lower-dimensional simplex. For example, the edges of a triangle are 1-faces (1-simplices), and its vertices are 0-faces (0-simplices). More generally, an $m$-face of a $k$-simplex is any $(m+1)$-subset of its vertex set, where $0 \leq m \leq k-1$.

A simplicial complex $K$ is a collection of simplices such that:
\begin{enumerate}
    \item Every face of a simplex in $K$ is also in $K$ (closure property).
    \item The intersection of any two simplices in $K$ is either empty or a common face of both.
\end{enumerate}

Thus, each simplex is an element of the simplicial complex, and the complex is constructed by ``gluing'' simplices together along their faces.

A $k$-chain is a formal sum of $k$-simplices with coefficients in a field (usually $\mathbb{Z}$ or $\mathbb{R}$). The space of $k$-chains is denoted by $C_k(K)$.

A $k$-cochain is a linear functional on $k$-chains, and the space of $k$-cochains is denoted by $C^k(K)$. These cochains represent discrete differential forms over the simplicial complex. A cochain can be seen as assigning a value to each $k$-simplex, much like how differential forms assign values to elements of a manifold.

\subsection{Boundary Operators and the Hodge Laplacian}

The fundamental connection between incidence matrices and algebraic topology emerges through the \emph{boundary operator}, denoted by $\partial_k: C_k(K) \to C_{k-1}(K)$, which maps $k$-simplices to $(k-1)$-simplices by summing over their boundaries. The dual \emph{coboundary operator} $\delta^k: C^k(K) \to C^{k+1}(K)$ acts on cochains and establishes the algebraic foundation for our centrality framework. These operators satisfy the fundamental boundary axiom $\partial_{k-1} \circ \partial_k = 0$, which ensures topological consistency and enables the construction of homology groups.

The boundary operator can be represented by the incidence matrix $B_k$ for $k$-simplices, establishing a direct correspondence between linear algebra and topology. For example, the incidence matrix $B_1$ between vertices and edges ($1$-simplices) encodes which edges are incident to which vertices, while $B_2$ encodes relationships between edges and triangles ($2$-simplices). More generally, for any $k$, the incidence matrix $B_k$ describes the boundary relationships between $(k-1)$-simplices and $k$-simplices, making it the discrete analog of the exterior derivative in differential geometry.

The Hodge Laplacian for simplicial complexes generalizes the graph Laplacian to higher-order structures and provides the spectral foundation for our centrality measures. For a given $k$, the Hodge Laplacian is defined as $\Delta_k = \delta_{k-1} \delta_{k-1}^* + \delta_k \delta_k^*$ \citep{lim2019hodgelaplaciansgraphs}, where $\delta_k^*$ denotes the adjoint of $\delta_k$ with respect to the standard Euclidean inner product on cochain spaces (hence, in the incidence-matrix representation, $\delta_k^*$ corresponds to matrix transpose). In terms of incidence matrices, this becomes $L_k = B_k^\top B_k + B_{k+1} B_{k+1}^\top$, where $B_k$ is the incidence matrix associated with the $k$-simplices, and $B_{k+1}$ is the incidence matrix associated with the $(k+1)$-simplices. The term $B_k^\top B_k$ represents the combinatorial analog of the divergence of the gradient (measuring local variation), while $B_{k+1} B_{k+1}^\top$ corresponds to the combinatorial curl (capturing circulation or rotational aspects).

\subsection{Incidence Matrices as Boundary Operators and Differential Structure}

The incidence matrix $B_k$ fundamentally encodes the boundary relationships between $k$-simplices and $(k-1)$-simplices, acting as the discrete boundary operator $\partial_k$. This matrix maps a $k$-simplex to its set of $(k-1)$-simplices (its boundary), providing the essential topological information needed for homological analysis. In graph theory, where we consider 1-dimensional simplicial complexes, the incidence matrix of edges and vertices precisely captures which vertices form the endpoints of each edge—exactly the information required for the boundary of an edge.

The directional encoding in our incidence matrix definition provides the necessary orientation that ensures the boundary axiom $\partial_1 \circ \partial_0 = 0$ holds rigorously. For any path of two consecutive edges sharing a vertex, the +1 and -1 contributions cancel out when applying consecutive boundary operations, which is precisely what the boundary axiom requires. This cancellation property is fundamental to establishing the connection between graph orientation and simplicial complex theory, ensuring that our framework respects topological constraints.

To understand the relationship between the Hodge Laplacian and familiar differential operators from vector calculus, we interpret the Laplacian terms through the lens of gradient, curl, and divergence operations. For 0-cochains (functions on vertices), the gradient is given by the incidence matrix $B_1$, while the divergence is represented by its adjoint $B_1^\top$. For 1-cochains (functions on edges), the curl is represented by the incidence matrix $B_2$, with its adjoint $B_2^\top$ capturing the corresponding co-curl operation. This correspondence makes the Hodge Laplacian a natural generalization of these operators to higher dimensions and discrete settings.

The decomposition $L_k = B_k^\top B_k + B_{k+1} B_{k+1}^\top$ reveals deep geometric structure: the term $B_k^\top B_k$ corresponds to the divergence of the gradient, capturing local variation at each $k$-simplex, while $B_{k+1} B_{k+1}^\top$ corresponds to the curl of higher-level cochains, capturing circulation or flux through the $k$-simplices. Thus, the Hodge Laplacian encodes both local and global topological information of the simplicial complex, extending the ideas of the Laplacian on graphs to higher-order structures and providing the spectral foundation that makes our SVD-based centrality measures both theoretically grounded and computationally tractable.

\section{Methods}
\label{sec:methods}

We derive centrality measures from the Singular Value Decomposition of the incidence matrix, establishing a framework grounded in Hodge theory and spectral graph analysis. This approach provides a unified mathematical foundation that simultaneously captures vertex and edge importance while naturally preserving directional information in network relationships.
\label{sec:why_svd_works}

\begin{figure}[ht]
  \centering
  \includegraphics[width=0.95\textwidth]{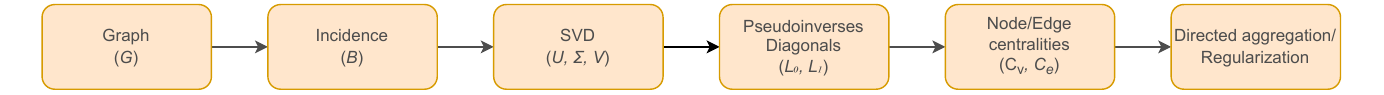}
  \caption{Our SVD-centrality calculation pipeline demonstrating the complete workflow from incidence matrix construction to centrality computation.}
  \label{fig:svd_pipeline}
\end{figure}

A matrix $B\in \mathbb{R}^{n \times m}$ can be decomposed using the \emph{Singular Value Decomposition} (SVD):
\[
B = U \Sigma V^T
\]
where $U\in \mathbb{R}^{n \times n}$ and $V \in \mathbb{R}^{m \times m}$ are orthogonal matrices, and $\Sigma\in \mathbb{R}^{n \times m}$ is a rectangular diagonal matrix containing the singular values (non-negative real values $\sigma$ with $Bv=\sigma u$, which are associated with left singular vectors $u\in \mathbb{R}^{n}$ and right singular vectors $v\in \mathbb{R}^{m}$). The singular value decomposition can be seen as a generalization of the eigendecomposition for a square symmetric matrix.

The SVD has been used to study graphs in different contexts, including community detection~\cite{sarkar2011community}, centrality~\cite{benzi2013ranking} and matrix functions~\cite{baglama2014analysis}. Most of these works, though, consider the SVD \emph{of the adjacency matrix}. Our contribution lies in applying the SVD specifically to the incidence matrix, which provides a perspective that simultaneously captures vertex-edge relationships and enables unified centrality analysis.

\Cref{fig:svd_pipeline} illustrates our complete computational workflow from graph representation to centrality calculation, demonstrating how the incidence matrix SVD enables simultaneous vertex and edge analysis through the pseudoinverse construction.

The Singular Value Decomposition (SVD) of the incidence matrix $B$ can be derived from the eigenvalue decompositions of $B^\top B$ and $B B^\top$, where $B \in \mathbb{R}^{n \times m}$, which arise in the context of the Hodge Laplacian. Specifically, $B^\top B$ relates to combinatorial divergence followed by gradient, and $B B^\top$ corresponds to gradient followed by divergence. The SVD expresses these relationships through singular values and singular vectors, providing insight into the graph's structure.

The eigenvalue decomposition of $B^\top B$ is:
\begin{equation}
B^\top B = V \Lambda V^\top,
\end{equation}
where $\Lambda$ is a diagonal matrix of eigenvalues of $B^\top B$, and $V$ contains the eigenvectors. Similarly, for $B B^\top$:
\begin{equation}
B B^\top = U \Lambda U^\top,
\end{equation}
where analogously $U$ contains the eigenvectors of $B B^\top$. The singular values, in the rectangular diagonal matrix $\Sigma$, where $\Sigma \Sigma ^T = \Lambda$, are the square roots of the eigenvalues of $B^\top B$ (and $B B^\top$):
$$ \sigma_i = \sqrt{\lambda_i}, \quad i = 1, 2, \ldots, \min(n, m), $$
including $\sigma_i = 0$ when $\lambda_i = 0$. These singular values are stored on the diagonal of the matrix $\Sigma$.

The matrices $U$ and $V$ contain the left and right singular vectors of $B$. The left singular vectors (columns of $U$) correspond to eigenvectors of $B B^\top$ and provide information about vertices in the graph. The right singular vectors (columns of $V$) correspond to eigenvectors of $B^\top B$, providing information about edges.

This structure reflects how $B$ acts as a linear transformation. The columns of $U$ span the vertex space, and the columns of $V$ span the edge space. The singular vectors form orthogonal bases for these spaces, highlighting key structural and topological features of the graph.

\subsection{Relationship with Hodge Laplacian}

We derive the matrices $U$ and $V$ from the singular value decomposition (SVD) of the incidence matrix $B$ using the spectral properties of the Hodge Laplacian. The general form of the $k$-th order Hodge Laplacian is $L_k = B_k^\top B_k + B_{k+1} B_{k+1}^\top$ \citep{horak2013spectra}. In particular, the Laplacian of order zero is $B_1 B_1^\top = L_0 \cong L = B B^\top$, where $B$ is the ordinary vertex--edge incidence matrix and $L$ is the graph Laplacian acting on vertices. The first-order Hodge Laplacian (the graph Helmholtzian) acts on edges and is given by $B_1^\top B_1 + B_{2} B_{2}^\top = L_1 \cong L_{\text{Hel}} = B^\top B + B_{\text{et}} B_{\text{et}}^\top$, where $B_{\text{et}}$ denotes the edge--triangle incidence matrix. In the absence of triangles, $B_2 = 0$, so $L_{\text{Hel}} \cong L_e = B^\top B$, which is the edge Laplacian.

These matrices are precisely those used for the spectral decomposition in the SVD (Section~\ref{sec:why_svd_works}): $B^\top B = L_e = V \Lambda V^\top$ and $B B^\top = L = U \Lambda U^\top$. The spectra of the Hodge Laplacians for vertices (graph Laplacian) and edges (graph Helmholtzian) have proven useful in various applications \citep{alain2024graphclassificationgaussianprocesses, chung1997spectral}. Both mathematical and experimental evidence indicate that the singular vectors from $U$ and $V$ serve as effective encoders of local and global graph information, capturing meaningful quantities such as centrality. For instance, the leading eigenvector of $B B^\top$ has already been employed for vertex centrality \citep{xu2021signless}, providing a precedent for our approach.

Consider a finite directed graph $G=(V,E)$ with incidence matrix $\mathbf{B}\in\mathbb{R}^{n\times m}$ and SVD $\mathbf{B}=\mathbf{U}\Sigma\mathbf{V}^\top$, where $\Sigma$ contains the singular values $\sigma_1\ge\cdots\ge\sigma_r>0$ and $r=\operatorname{rank}(\mathbf{B})$. The spectral decomposition reveals that the Hodge Laplacians $\mathbf{L}_0=\mathbf{B}\mathbf{B}^\top$ and $\mathbf{L}_1=\mathbf{B}^\top\mathbf{B}$ satisfy $\mathbf{L}_0=\mathbf{U}\Sigma^2\mathbf{U}^\top$ and $\mathbf{L}_1=\mathbf{V}\Sigma^2\mathbf{V}^\top$. Thus, the nonzero eigenvalues of both $\mathbf{L}_0$ and $\mathbf{L}_1$ are exactly the squared singular values $\{\sigma_i^2\}_{i=1}^r$, with eigenvectors given by the corresponding columns of $\mathbf{U}$ and $\mathbf{V}$, respectively.

The topological meaning of the zero eigenvalues follows from the rank--nullity theorem. Since $\operatorname{rank}(\mathbf{B})=n-c$, where $c$ is the number of connected components, we have $\dim\ker(\mathbf{B}^\top)=c$ and $\dim\ker(\mathbf{B})=m-(n-c)=\beta_1$, with $\beta_1$ being the first Betti number (the cycle space dimension). These kernel dimensions correspond to the multiplicities of the zero eigenvalues in $\mathbf{L}_0$ and $\mathbf{L}_1$, establishing a direct link between spectral properties and network topology \citep{biggs1993algebraic, friedman1998computing}. 

An important invariance property ensures robustness: changing edge orientations by transforming $\mathbf{B}$ to $\mathbf{B}' = \mathbf{B}\mathbf{D}$, where $\mathbf{D}$ is a diagonal matrix with $\pm 1$ entries, leaves the singular values unchanged due to the unitary invariance of the SVD. Only the signs of the right singular vectors are affected, leaving $\Sigma$ and $\Sigma^2$ invariant, and hence the spectra of $\mathbf{L}_0$ and $\mathbf{L}_1$ are preserved.

This spectral analysis leads naturally to the discrete Hodge decomposition \citep{eckmann1945harmonische, horak2013spectra}. The SVD induces an orthogonal splitting of the vertex space $\mathbb{R}^n$ and the edge space $\mathbb{R}^m$ into exact, co-exact, and harmonic subspaces, corresponding respectively to the images and kernels of $\mathbf{B}$ and $\mathbf{B}^\top$. This decomposition forms the theoretical foundation for our centrality measures, as it elucidates how information flows and potentials distribute across the topological structure of the network.

\subsection{Hodge Decomposition and the Fundamental Subspaces}

The Hodge decomposition offers a tool for analyzing vector fields on graphs and higher-dimensional manifolds. In simplicial complexes, this decomposition allows the feature space $\mathbb{R}^{N_k}$ (for $k$-simplices) to be split into three orthogonal components: exact, co-exact, and harmonic. Mathematically, this is expressed as:
$$
\mathbb{R}^{N_k} = \text{Im}(B_{k+1}) \oplus \text{Im}(B_k^\top) \oplus \ker(L_k),
$$
where $\text{Im}(B_{k+1})$ represents the exact subspace, $\text{Im}(B_k^\top)$ the co-exact subspace, and $\ker(L_k)$ the harmonic subspace. The dimension of the harmonic subspace is given by the $k$-th Betti number, indicating the number of $k$-dimensional holes.

When we express this decomposition for the vertices and edges, $\mathbb{R}^n$ and $\mathbb{R}^m$ respectively, we have:
$$\mathbb{R}^n = \text{Im}(B) \oplus \ker (L) =  \text{Im}(B) \oplus \ker (B B^\top), $$
   $$\mathbb{R}^m = \text{Im}(B_2) \oplus \text{Im}(B^\top) \oplus \ker(L_1). $$
Since we are not considering triangles, i.e, $B_2=0$ and $L_1 = L_e=B^\top B$, we have:
$$\mathbb{R}^m = \text{Im}(B^\top) \oplus \ker(L_e)=  \text{Im}(B^\top) \oplus \ker(B^\top B). $$
As $\ker(B^\top B) = \ker(B)$ and $\ker(B B^\top) = \ker(B^T)$, then we simply have:
$$\mathbb{R}^n =   \text{Im}(B) \oplus \ker (B^\top), $$
$$\mathbb{R}^m =  \text{Im}(B^\top) \oplus \ker(B). $$

This is exactly what we obtain from the SVD decomposition through the Four Fundamental Subspaces theorem. Thus, the SVD decomposition can be as effective as the Hodge Decomposition in capturing the essential structural properties of the graph.

\paragraph{Conclusion} The SVD of the incidence matrix $B$ simultaneously diagonalizes the Hodge Laplacians $L_0$ and $L_1$ and provides orthogonal `modal' bases for vertices/vertex-space and edges/edge-space.

\subsection{Notation and modal (Fourier) intuition}

This section presents a \emph{derivation from first principles} of node and edge centralities built from the singular value decomposition (SVD) of the incidence matrix. We do \emph{not} introduce a heuristic: the measures below are canonical spectral objects with precise variational and electrical interpretations.  The exposition proceeds as follows.  First we introduce notation and modal intuition (Fourier/Rayleigh viewpoint).  Second we give compact, provable definitions of vertex and edge centrality via direct SVD formulation.  Third we prove the core identities (SVD expressions, resistance relations) and establish invariance under orientation (hence the need for directed aggregation).  Fourth we give computational and regularization recipes for practical implementation.  Wherever useful we link the SVD centrality to classical current-flow (electrical) closeness and to related network invariants.

Let $G=(V,E)$ be a (possibly directed) graph with $n=|V|$ vertices and $m=|E|$ (oriented) edges.  Let $B\in\mathbb R^{n\times m}$ denote the oriented incidence matrix (each column has a $+1$ and a $-1$ at endpoints of the corresponding oriented edge).  Write the compact SVD
$$
B = U \Sigma V^\top,
$$
with $U\in\mathbb R^{n\times r}$, $V\in\mathbb R^{m\times r}$, $\Sigma=\operatorname{diag}(\sigma_1,\dots,\sigma_r)$, $\sigma_1\ge\sigma_2\ge\cdots\ge\sigma_r\ge0$, and $r=\operatorname{rank}(B)$.  Define the Hodge (vertex and edge) Laplacians
$$
L_0 = B B^\top \in\mathbb R^{n\times n},\qquad L_1 = B^\top B \in\mathbb R^{m\times m},
$$
so that
$$
L_0 = U\Sigma^2 U^\top,\qquad L_1 = V\Sigma^2 V^\top.
$$

The SVD provides orthogonal modal bases for vertex potentials (columns of $U$) and edge flows (columns of $V$).  The identities
$$
B v_k = \sigma_k u_k,\qquad B^\top u_k = \sigma_k v_k
$$
show that $u_k$ and $v_k$ are eigenvectors of $L_0$ and $L_1$ with eigenvalue $\sigma_k^2$.  The Rayleigh viewpoint says that $\sigma_k^2$ is the energy coupling of mode $k$: if $\phi=u_k$ then the induced flow $f=B^\top\phi$ satisfies $\|f\|_2^2=\phi^\top L_0\phi=\sigma_k^2\|\phi\|_2^2$.  Small $\sigma_k$ correspond to global, slowly-varying modes; large $\sigma_k$ correspond to energetic, flow-generating modes.  This modal/Fourier intuition motivates but does not replace the rigorous pseudoinverse definitions that follow: the pseudoinverse weights $1/\sigma_k^2$ emphasize global low-frequency modes, which is precisely the desired behavior for a resistance-based global centrality.

\subsection{Definitions: SVD centrality measures}
\begin{definition}[Vertex SVD centrality]
Given the SVD $B = U \Sigma V^\top$ with singular values $\sigma_1 \geq \sigma_2 \geq \ldots \geq \sigma_r > 0$, the \emph{vertex SVD centrality} of node $i$ is
$$
C_v(i) := \sum_{k=1}^{r} \frac{u_{k,i}^2}{\sigma_k^2}.
$$
\end{definition}

\begin{definition}[Edge SVD centrality]
Similarly, the \emph{edge SVD centrality} of edge $e$ is
$$
C_e(e) := \sum_{k=1}^{r} \frac{v_{k,e}^2}{\sigma_k^2}
$$
where $v_{k,e}$ are the components of the right singular vectors from the SVD of the regularized matrix $M$.
\end{definition}

These definitions are coordinate-free, invariant under orthonormal basis changes in degenerate eigenspaces, and admit immediate electrical/variational meanings (next subsection). Our experimental implementation employs a regularized variant of these measures for enhanced numerical stability, as detailed in the computational implementation section.

\subsection{Algebraic identities and connection to current-flow and effective resistance}
The \Cref{thm:effective_resistance} collects the elementary but essential equalities and the link to effective resistance and current-flow closeness as we discuss next.

The connection between SVD-based centrality and electrical network theory emerges through the pseudoinverse relationship and effective resistance calculations. For a connected undirected graph with incidence matrix $B$ having compact SVD $B=U\Sigma V^\top$, the pseudoinverses of the Hodge Laplacians $L_0=BB^\top$ and $L_1=B^\top B$ take the explicit forms $L_0^+ = U\Sigma^{-2}U^\top$ and $L_1^+ = V\Sigma^{-2}V^\top$, where $\Sigma^{-2}$ inverts positive singular values while preserving zeros.

The diagonal elements of these pseudoinverses, which define our centrality measures, can be expressed in terms of the SVD components as
\begin{equation}
[L_0^+]_{ii}=\sum_{\sigma_k>0} \frac{u_{k,i}^2}{\sigma_k^2} \quad \text{and} \quad [L_1^+]_{ee}=\sum_{\sigma_k>0} \frac{v_{k,e}^2}{\sigma_k^2}.
\label{eq:centrality_definitions}
\end{equation}
These expressions reveal how centrality values depend on both the singular value spectrum and the localization of singular vectors at specific vertices or edges.

The fundamental relationship with electrical network theory becomes apparent through effective resistance calculations. For any pair of nodes $i,j$, the effective resistance is given by the quadratic form \citep{klein1993resistance}
\begin{equation}
R_{ij} = (e_i-e_j)^\top L_0^+ (e_i-e_j) = [L_0^+]_{ii} + [L_0^+]_{jj} - 2[L_0^+]_{ij}.
\label{eq:effective_resistance}
\end{equation}
By summing over all target nodes $j$, we obtain a relationship that connects individual vertex centrality to global network accessibility.

Expanding this sum using the SVD representation yields for any node $i$:
\[
\sum_{j=1}^n R_{ij} = \sum_{j=1}^n \sum_{k=1}^r \frac{(u_{k,i} - u_{k,j})^2}{\sigma_k^2}.
\]
The key insight comes from exploiting the orthonormality of the left singular vectors and the fact that $U^\top \mathbf{1} = 0$ for the graph Laplacian. After algebraic manipulation (see \Cref{thm:effective_resistance} in \Cref{app:proof}), this yields:
\[
\sum_{j=1}^n R_{ij} = n \sum_{k=1}^r \frac{u_{k,i}^2}{\sigma_k^2} + \sum_{k=1}^r \frac{1}{\sigma_k^2} = n[L_0^+]_{ii} + \operatorname{tr}(L_0^+).
\]

This relationship establishes a theoretical connection between SVD vertex centrality and current-flow closeness centrality for undirected graphs~\cite{brandes2005centrality}. Since current-flow closeness is defined as $\mathrm{CFCloseness}(i) := (n-1)/\sum_{j\ne i} R_{ij}$ and $\sum_{j\ne i} R_{ij} = \operatorname{tr}(L_0^+) - [L_0^+]_{ii}$, ranking vertices by increasing $[L_0^+]_{ii}$ corresponds to ranking by decreasing current-flow closeness centrality in undirected connected graphs. This provides a theoretical foundation for interpreting SVD vertex centrality $C_v(i) = [L_0^+]_{ii}$ as measuring electrical resistance in undirected networks. However, this equivalence fundamentally breaks down for directed graphs because current-flow closeness centrality algorithms convert directed graphs to undirected form, discarding directional information that SVD centrality preserves through the oriented incidence matrix as we further detail in \Cref{sec:theoretical_properties_implementation}.

The breakdown of equivalence in directed graphs motivates the primary contribution of this work: SVD centrality provides a principled approach for directed network analysis without losing directional structure. Unlike betweenness centrality, which suffers from implementation-dependent tie-breaking and sparse results in directed graphs, SVD centrality leverages the spectral properties of the oriented incidence matrix to provide dense, mathematically consistent rankings for all vertices and edges.

\paragraph{Physical Interpretation and Network Theory Implications.} The SVD-effective resistance relationship in undirected graphs reveals that our centrality framework reflects fundamental physical principles governing flow in networked systems. For directed networks, the oriented incidence matrix captures asymmetric flow patterns that traditional undirected measures cannot represent, making SVD centrality particularly valuable for applications requiring directional sensitivity.

\subsection{Extension to Hypergraphs}
\label{sec:hypergraph_extension_methods}

The SVD centrality framework naturally extends to hypergraphs, where edges (hyperedges) can connect multiple vertices simultaneously. For a hypergraph $H = (V, E)$ with vertex set $V$ and hyperedge set $E$, we define the node-hyperedge incidence matrix $B \in \mathbb{R}^{|V| \times |E|}$ following the standard hypergraph representation where:
$$
B_{ih} = \begin{cases}
1 & \text{if vertex } i \text{ belongs to hyperedge } h \\
0 & \text{otherwise}
\end{cases}
$$

The hypergraph Hodge Laplacians are then defined as:
\begin{align}
L_0^{(H)} &= BB^\top \quad \text{(vertex Laplacian)} \\
L_1^{(H)} &= B^\top B \quad \text{(hyperedge Laplacian)}
\end{align}

\begin{definition}[Hypergraph SVD centrality]
\label{def:hypergraph_svd_centrality}
For a hypergraph with incidence matrix $B = U\Sigma V^\top$, the vertex and hyperedge centralities are:
\begin{align}
C_v^{(H)}(i) &= [L_0^{(H)+}]_{ii} = \sum_{k:\sigma_k>0} \frac{u_{k,i}^2}{\sigma_k^2} \label{eq:hypergraph_vertex_centrality}\\
C_e^{(H)}(h) &= [L_1^{(H)+}]_{hh} = \sum_{k:\sigma_k>0} \frac{v_{k,h}^2}{\sigma_k^2} \label{eq:hypergraph_edge_centrality}
\end{align}
\end{definition}

This generalization proves particularly valuable for complex systems with multiway interactions: ecological pollination networks, epidemiological contact patterns, scientific collaborations, and legislative co-sponsorship structures. Experimental validation on real-world hypergraph datasets (\S\ref{sec:hypergraph_experiments}) demonstrates the framework's effectiveness, revealing a fundamental advantage over traditional measures: SVD centrality provides dense rankings that avoid the sparsity issues characteristic of betweenness centrality in hypergraph structures.

\subsection{Theoretical Properties and Implementation}
\label{sec:theoretical_properties_implementation}

\paragraph{Hub and Authority Centralities for Directed Networks.} A key insight is that vertex centrality $C_v$ is orientation-invariant, making directed networks require an \emph{aggregation step} to produce orientation-aware node scores. We accomplish this in the most principled way: compute orientation-sensitive edge centralities $C_e(e)$, then aggregate incoming/outgoing edge scores to nodes.\begin{definition}[Authority and hub centralities (aggregation)]
Let $\mathcal I_{\mathrm{in}}\in\{0,1\}^{n\times m}$ be the incidence indicator for incoming edges and $\mathcal I_{\mathrm{out}}$ for outgoing edges.  Define
$$
\mathbf c_{\mathrm{auth}} \;=\; \mathcal I_{\mathrm{in}}\,\mathbf c_e,\qquad
\mathbf c_{\mathrm{hub}} \;=\; \mathcal I_{\mathrm{out}}\,\mathbf c_e,
$$
where $\mathbf c_e$ is the vector with components $C_e(e) = \sum_{k=1}^{r} v_{k,e}^2/\sigma_k^2$.  Optionally, to incorporate intrinsic vertex accessibility, use a convex combination
$$
\mathbf c_{\mathrm{auth}}^\alpha \;=\; \alpha\,\mathbf c_v + (1-\alpha)\,\mathcal I_{\mathrm{in}}\,\mathbf c_e,\qquad
\mathbf c_{\mathrm{hub}}^\alpha \;=\; \alpha\,\mathbf c_v + (1-\alpha)\,\mathcal I_{\mathrm{out}}\,\mathbf c_e,
$$
with $\alpha\in[0,1]$.
\end{definition}

The convex-combination variant yields a family of orientation-aware scores interpolating between pure global accessibility ($\alpha=1$) and pure edge-driven orientation-aware scores ($\alpha=0$). This convex form is useful when one wants to temper edge-aggregation by vertex-level accessibility.

\paragraph{Orientation Invariance Property.}
A key structural property ensures robustness of the spectral framework: the Laplacians $L_0=BB^\top$ and $L_1=B^\top B$ are invariant in the sense relevant to our centralities under edge-orientation flips. This orientation invariance provides formal justification that $C_v$ is intrinsically undirected and validates our approach of building orientation-aware node scores by aggregating oriented edge centralities.

The orientation invariance property can be stated precisely: if $\widetilde B$ is obtained from $B$ by flipping the orientation of any subset of edges (multiplying the corresponding columns by $-1$), then $BB^\top = \widetilde B\widetilde B^\top$, while $\widetilde B^\top \widetilde B$ is related to $B^\top B$ by a diagonal sign conjugation, $\widetilde B^\top \widetilde B = D\,(B^\top B)\,D$, where $D$ is diagonal with entries in $\{\pm1\}$. Consequently, $L_0^+$ is unchanged by orientation flips, $L_1$ and $\widetilde L_1$ share the same spectrum, and the diagonal centralities $C_v$ and $C_e$ remain invariant under edge-orientation changes.

This orientation invariance follows directly from the matrix algebra: flipping orientation multiplies column $j$ of $B$ by $-1$, transforming $b_j\mapsto -b_j$. Since $BB^\top = \sum_j b_j b_j^\top$ and $\widetilde B\widetilde B^\top = \sum_j (\pm b_j)(\pm b_j)^\top = \sum_j b_j b_j^\top$, we have $BB^\top=\widetilde B\widetilde B^\top$. For the edge Gram, however, one obtains $\widetilde B^\top\widetilde B = D\,(B^\top B)\,D$, so the two matrices are not entrywise identical but are congruent and therefore share the same spectrum. The Moore--Penrose pseudoinverse transforms accordingly as $(\widetilde B^\top\widetilde B)^+ = D\,(B^\top B)^+\,D$, and because $D^2=I$ the diagonal entries used to define the edge centrality satisfy $\operatorname{diag}\big((\widetilde B^\top\widetilde B)^+\big)=\operatorname{diag}\big((B^\top B)^+\big)$.

\paragraph{Truncated SVD}
\label{par:truncated_svd}
The centrality measures $C_v(i)$ and $C_e(e)$ are defined as sums over the non-zero singular values of the incidence matrix. Mathematically, this corresponds to using the compact SVD $B = U_r \Sigma_r V_r^\top$, where $U_r \in \mathbb{R}^{n \times r}$, $V_r \in \mathbb{R}^{m \times r}$, and $\Sigma_r = \operatorname{diag}(\sigma_1, \dots, \sigma_r)$ contains the $r = \operatorname{rank}(B)$ positive singular values. The pseudoinverses of the Hodge Laplacians are then given by:
\[
L_0^+ = U_r \Sigma_r^{-2} U_r^\top, \quad
L_1^+ = V_r \Sigma_r^{-2} V_r^\top.
\]
This formulation explicitly uses only the non-zero singular values, avoiding division by zero and ensuring numerical stability. The truncation is not arbitrary but is grounded in the algebraic structure of the incidence matrix: the zero singular values correspond to the harmonic subspace (kernel of $L_0$ and $L_1$), which represents global topological invariants (connected components and cycles) that do not contribute to the diffusive flow captured by centrality. By discarding these zero modes, we focus precisely on the modes that govern information propagation and resistance distances. Moreover, in practical computations, we may further truncate to the largest $k < r$ singular values to approximate the centrality efficiently while preserving the dominant structural patterns, as the smallest positive singular values contribute minimally to the centrality sums due to the $1/\sigma_k^2$ weighting. The quality of such low-rank approximations can be measured by the Frobenius norm of the reconstruction error: $\|B - U_k \Sigma_k V_k^\top\|_F = \sqrt{\sum_{i=k+1}^{r} \sigma_i^2}$, which quantifies the spectral energy discarded by truncation. This truncated approach is especially valuable for large-scale networks, where full SVD is computationally prohibitive.

\paragraph{Regularization for Numerical Stability.} Practical implementations require regularization for numerical stability when working with directed networks where the incidence matrix may have small singular values. Our experimental implementation uses matrix-level regularization, replacing the incidence matrix $B$ with:
$$M = \sqrt{\lambda} \, B + \sqrt{1-\lambda} \, I_{n \times m}$$
where $I_{n \times m}$ is an appropriately dimensioned identity-like matrix and $\lambda \in [0,1]$ interpolates between original structure ($\lambda = 1$) and regularization ($\lambda = 0$). The centrality measures follow the same spectral formulation:
\begin{align}
C_v(i) &= \sum_{k=1}^{r} \frac{u_{M,k,i}^2}{\sigma_{M,k}^2}, \qquad
C_e(e) = \sum_{k=1}^{r} \frac{v_{M,k,e}^2}{\sigma_{M,k}^2}
\end{align}
using the SVD $M = U_M \Sigma_M V_M^\top$.

\paragraph{Normalized Scores for Ranking.} For applications preferring ``larger is better' scores, the SVD centrality values can be inverted and normalized:
$$
S_v(i) = \frac{1/(C_v(i) + \tau)}{\max_\ell 1/(C_v(\ell)+\tau)}, \qquad
S_e(e) = \frac{1/(C_e(e) + \tau)}{\max_{\ell} 1/(C_e(\ell)+\tau)},
$$
with small $\tau>0$ (e.g., $\tau=10^{-8}$) to avoid division by zero. These $S_v,S_e\in(0,1]$ are convenient for visualization and ranking. In our experimental implementation, we use the normalized scores in the aggregation step, computing $\mathbf s_{\mathrm{auth}}^\alpha = \alpha\,\mathbf s_v + (1-\alpha)\,\mathcal I_{\mathrm{in}}\,\mathbf s_e$ and $\mathbf s_{\mathrm{hub}}^\alpha = \alpha\,\mathbf s_v + (1-\alpha)\,\mathcal I_{\mathrm{out}}\,\mathbf s_e$, where $\mathbf s_v$ and $\mathbf s_e$ are the vectors with components $S_v(i)$ and $S_e(e)$ respectively, and subsequently normalize the resulting hub and authority scores ($S_{auth}, S_{hub}$) by their respective maximum values to ensure they lie in $[0,1]$.

\FloatBarrier
\section{Experimental Validation}
\label{sec:experimental_results}

Our experimental validation follows a systematic progression: first establishing theoretical equivalence for undirected networks, then demonstrating the unique advantages for directed networks where traditional measures face limitations, and finally showing practical applications across diverse real-world network domains. This progression reveals how SVD centrality bridges the gap between spectral graph theory and practical network analysis for directed systems.

\subsection{Equivalence Validation}

We begin our experimental validation by confirming the theoretical equivalence between SVD vertex centrality and current-flow closeness centrality in undirected networks. This validation establishes the physical grounding of our approach before demonstrating its unique capabilities for directed networks. By comparing against Current-Flow Closeness—which is mathematically defined via effective resistance—we verify that our spectral decomposition correctly captures the diffusive properties of the network.

Table~\ref{tab:equivalence_validation} presents the correlation analysis across three distinct undirected network topologies chosen to span different structural regimes: social community structure, stochastic connectivity, and regular lattice geometry.

\begin{table}[h]
\centering
\caption{Pearson correlation coefficients between SVD vertex centrality and Current-Flow Closeness. The consistently high correlations across heterogeneous structures validate the theoretical connection to electrical resistance principles.}
\label{tab:equivalence_validation}
\begin{tabular}{lccc}
\toprule
Network & Nodes & Baseline Measure & Pearson $\rho$ \\
\midrule
Zachary's Karate Club & 34 & Current-Flow Closeness & 0.928 \\
Random ER(15, 0.3) & 15 & Current-Flow Closeness & 0.975 \\
Path Graph $P_8$ & 8 & Current-Flow Closeness & 0.986 \\
\bottomrule
\end{tabular}
\end{table}

The benchmark dataset is Zachary's Karate Club \citep{zachary1977information}, a canonical social network ($N=34$, $M=78$) exhibiting strong community structure and non-trivial hierarchy. As shown in Figure~\ref{fig:zachary_comparison}, the SVD centrality (left) and current-flow closeness (right) produce strikingly similar distributions ($\rho = 0.928$). Both measures correctly identify the administrator (Node 1) and the instructor (Node 34) as the primary structural poles—clearly highlighted in Figure~\ref{fig:zachary_comparison} as the two main poles of the network—reflecting their roles as centers of gravity in the social force field. This agreement confirms that the dominant singular vectors of the incidence matrix encode the same global accessibility information as the inverse Laplacian kernel used in random walk definitions.

\begin{figure}[htbp]
\centering
\includegraphics[width=\textwidth]{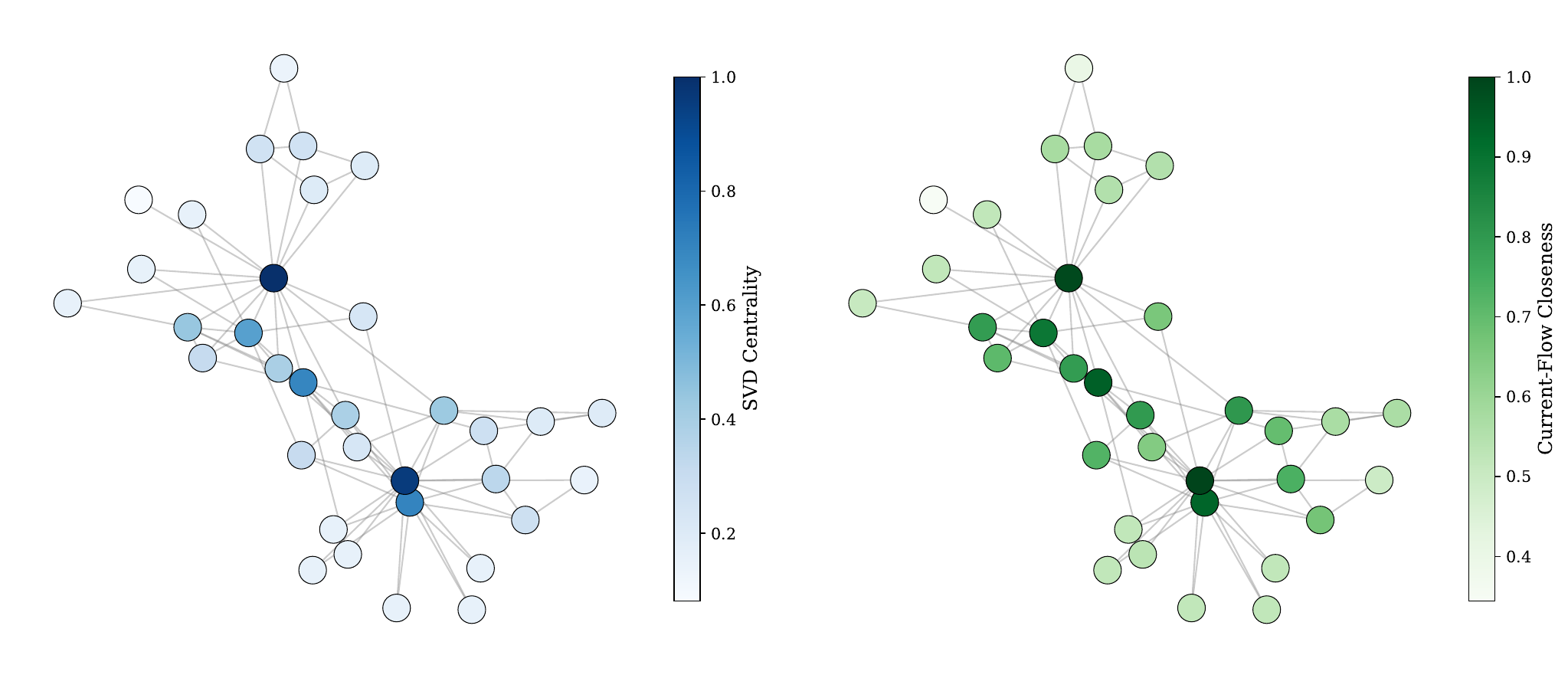}
\caption{Visual comparison of SVD vertex centrality (left) and current-flow closeness centrality (right) on Zachary's karate club network. The strong correlation ($\rho = 0.928$) and identical structural highlighting confirm that SVD centrality effectively captures electrical resistance principles in undirected networks.}
\label{fig:zachary_comparison}
\end{figure}

The analysis of synthetic graphs reveals further insights into the method's geometric fidelity. The Path Graph $P_8$, representing a 1D regular lattice, yields the highest correlation ($\rho = 0.986$). In this simple geometry, effective resistance scales linearly with geodesic distance, and the SVD's harmonic modes align perfectly with the standing waves of the 1D Laplacian. The Erdős-Rényi graph ($N=15, p=0.3$), representing a stochastic regime with short average path lengths, also shows exceptional agreement ($\rho = 0.975$). This suggests that our spectral formulation is robust across both regular, high-diameter lattices and random, small-world topologies.

It is important to address why the correlations are not exactly unity ($\rho = 1.0$) despite the theoretical equivalence. The deviations arise from necessary computational distinctions. Our implementation employs the regularization $M = \sqrt{\lambda} B + \sqrt{1-\lambda} I$, with $\lambda=0.99$ and $\tau=10^{-8}$, to ensure numerical stability and to handle potential rank-deficiencies in disconnected components—a common issue in real-world data. Standard current-flow solvers, by contrast, typically operate on the unregularized Laplacian pseudoinverse. These regularization terms introduce a small ``spectral bias'' that favors stability but creates minor numerical residuals compared to the raw inverse. However, the strong linear trends shown in Figure~\ref{fig:svd_vs_closeness_comparison} confirm that the fundamental ranking principle—minimizing the effective resistance distance to all other nodes—is preserved.

\begin{figure}[htbp]
\centering
\includegraphics[width=0.95\textwidth]{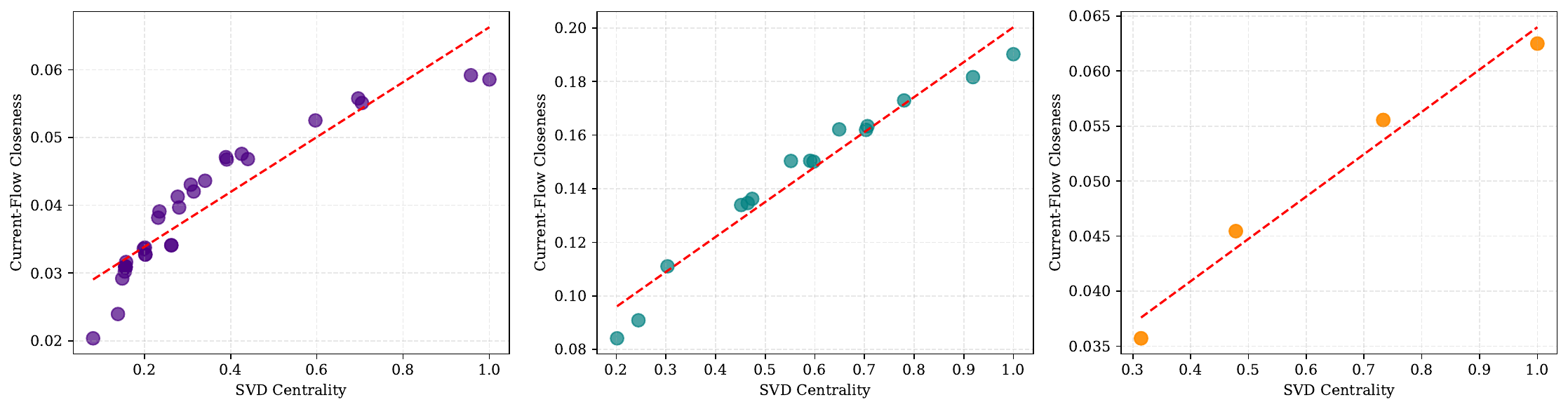}
\caption{Correlation analysis between SVD centrality and Current-Flow Closeness Centrality. Left: Zachary's Karate Club ($\rho=0.928$). Center: Random ER Graph ($\rho=0.975$). Right: Path Graph $P_8$ ($\rho=0.986$). The tight linear relationships across all topologies confirm the method's theoretical validity.}
\label{fig:svd_vs_closeness_comparison}
\end{figure}

Physically, these results demonstrate that SVD centrality can be interpreted as a measure of ``diffusive accessibility'. Just as current-flow closeness measures how easily current spreads from a node to the rest of the network (or conversely, how easily a random walker reaches the node), SVD centrality captures this global connectivity through the energy spectrum of the incidence matrix. This establishes a rigorous physical baseline for the method before we extend it to directed networks, where the notion of reversible random walks breaks down and our orientation-preserving spectral approach becomes essential.
These experimental results validate our theoretical foundation while establishing the baseline for directed network innovations where traditional measures may encounter challenges in preserving directional structure.

\subsection{Controlled Grid Experiment}
\label{sec:controlled-grid}

We evaluate the directional sensitivity of our spectral definitions on a deliberately simple but revealing toy model: a $4\times4$ directed grid (16 nodes) with an embedded hub–authority motif. In the construction used for the figures, the \emph{hub} at grid coordinate $(2,2)$ has directed edges \emph{to} all other nodes in the graph (high out-degree, low in-degree), while the \emph{authority} at $(2,3)$ has directed edges \emph{from} all other nodes (high in-degree, low out-degree). This creates an unambiguous test: the hub is a primary broadcaster, the authority is a primary receiver, and it is desirable that a directional centrality measure should identify these roles clearly.

Figure~\ref{fig:grid_raw_svd_vs_betweenness} juxtaposes three visualizations: SVD node scores \(S_v\) (left), SVD edge scores \(S_e\) (center), and classical edge betweenness (right). The spectral edge centrality \(S_e\) in the center panel successfully identifies the critical structural edges: the outgoing edge from the hub $(2,2)$ that connects to the authority $(2,3)$ receives the highest score, reflecting its role in the directed flow structure as the most influential edge, since it connects the two key nodes. Also, these two most influential nodes are the ones identified by SVD node score $S_v$. Edge betweenness, by contrast, fails to highlight these structurally decisive edges. Because the graph is highly connected with numerous equivalent shortest paths between most node pairs, betweenness distributes values according to algorithmic tie-breaking rather than recognizing the fundamental directional roles encoded in the hub-authority motif.

\begin{figure}[ht]
  \centering
  \includegraphics[width=\linewidth]{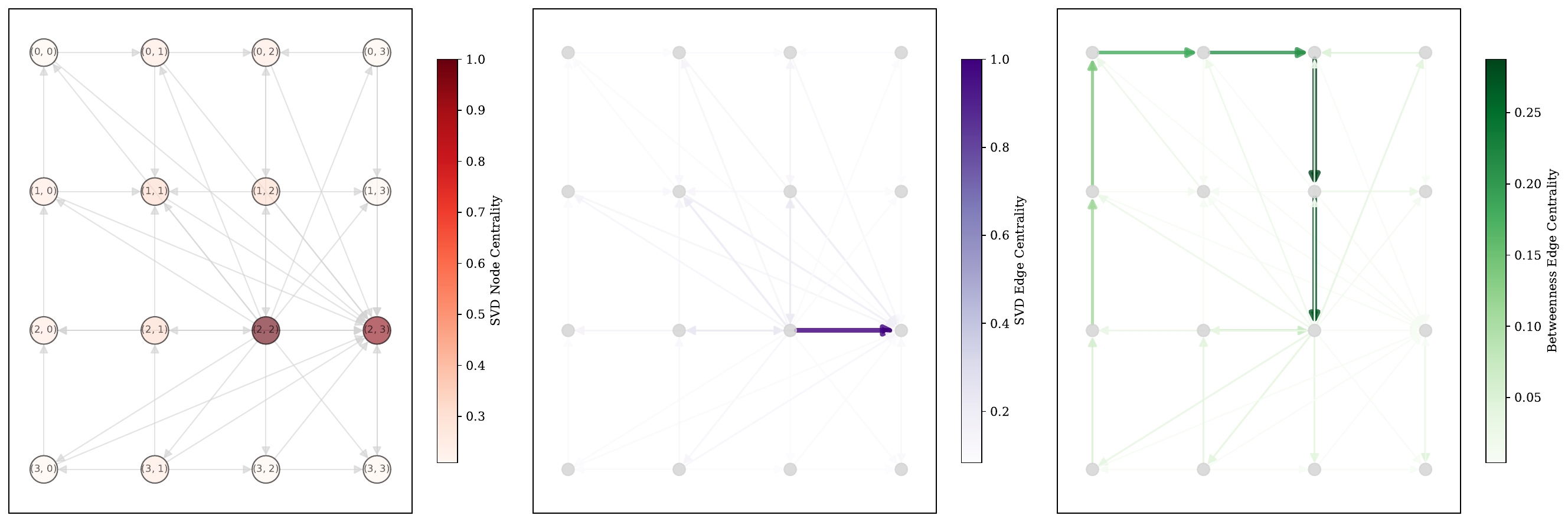}
  \caption{Raw SVD node centrality (left), SVD edge centrality (center) and edge betweenness (right) on the $4\times4$ directed grid. The SVD edge measure highlights edges connected to the planted hub at $(2,2)$ (outgoing) and authority at $(2,3)$ (incoming); betweenness distributes values uniformly due to path multiplicity and fails to identify the directional roles.}
  \label{fig:grid_raw_svd_vs_betweenness}
\end{figure}

The limitations of betweenness in this setting are worth noting. In a highly connected directed graph where most nodes can reach each other through multiple shortest paths of equal length, betweenness tends to assign similar or low values to many edges, or distributes credit among symmetric alternatives based on tie-breaking rules. As a result, the measure may appear relatively uniform or less informative about which nodes or edges play key directional roles. This does not mean betweenness is incorrect—it answers the question of which edges participate in shortest paths—but this perspective may not always highlight broadcasters and receivers in densely connected directed networks.

Figure~\ref{fig:grid_hub_authority} demonstrates the core advantage of the spectral approach at the node level. The SVD hub centrality (left panel) concentrates sharply at $(2,2)$, correctly identifying the broadcaster that sends to all other nodes. The SVD authority centrality (center panel) concentrates sharply at $(2,3)$, correctly identifying the receiver that accepts from all other nodes. These measures successfully recover the planted directional roles with minimal ambiguity. Node betweenness (right panel), by contrast, shows a relatively uniform distribution across the grid, failing to distinguish the hub and authority from other nodes. This uniform pattern arises because, in this highly connected structure, many nodes participate equally in shortest paths due to the abundance of equivalent routes.

\begin{figure}[ht]
  \centering
  \includegraphics[width=\linewidth]{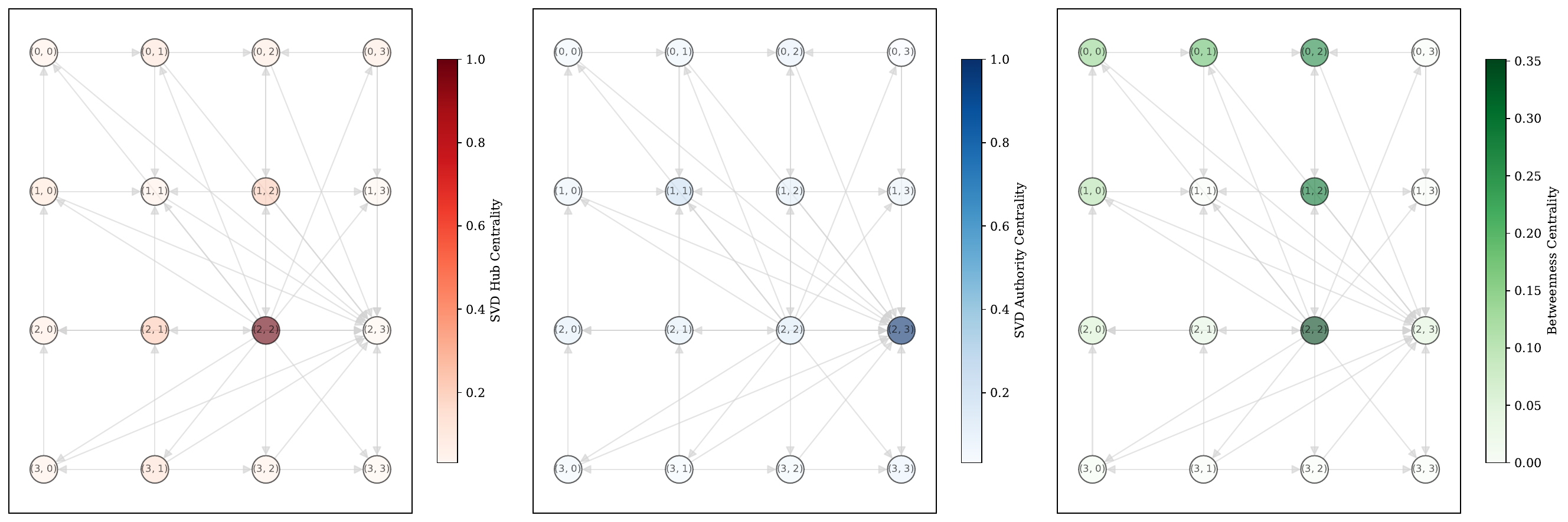}
  \caption{SVD hub centrality (left), SVD authority centrality (center) and node betweenness (right) on the $4\times4$ directed grid. The SVD node measures correctly identify the planted broadcaster $(2,2)$ (with edges to all nodes) and receiver $(2,3)$ (with edges from all nodes), while betweenness shows near-uniform values and fails to distinguish these directional roles.}
  \label{fig:grid_hub_authority}
\end{figure}

The controlled grid therefore reveals a fundamental difference in what these measures detect. Spectral centralities identify \emph{directional structural roles}: they recognize that $(2,2)$ is a broadcaster and $(2,3)$ is a receiver based on the oriented edge pattern, regardless of shortest-path considerations. Betweenness identifies \emph{shortest-path bottlenecks}, which in densely connected graphs with high path multiplicity means it often produces uninformative, near-uniform distributions that obscure directional roles.

Two practical conclusions follow directly. First, if the analytic goal is to identify broadcaster/receiver roles or to find nodes and edges that are structurally important for directed flow (influence propagation, information diffusion, directed transport), spectral centralities provide clear, interpretable answers even in dense networks. Second, betweenness remains appropriate for questions about shortest-path traffic under a specified cost metric, but in settings with substantial path multiplicity—common in many real-world directed networks—it may fail to identify the directional roles that matter most for understanding oriented dynamics. The grid experiment makes this trade-off transparent: spectral centralities detect role-based importance; betweenness detects geodesic routing, which can be ambiguous or uninformative when many equivalent routes exist.

\subsection{Real-World Network Analysis}

We now validate practical applicability across diverse real-world domains. The central question: \textit{How does our spectral framework perform when applied to the complex, heterogeneous networks that define real-world systems?}

This validation spans social networks (revealing influence patterns), biological systems (capturing regulatory hierarchies), and infrastructure networks (identifying critical connections). Each domain tests different aspects of our framework while maintaining focus on directed network analysis where traditional measures struggle.

\subsubsection{Social Networks: Dutch School Friendships}

The Dutch school friendship network \citep{snijders2010introduction} serves as an ideal substrate to test the resolving power of our spectral framework against established measures. We analyze the third wave of this longitudinal study ($N=25$ students, $M=133$ directed edges), a phase where the social hierarchy has stabilized into a dense, reciprocal structure with non-trivial community overlaps.

The comparative analysis in \Cref{fig:dutch_general} and Table~\ref{tab:dutch_general_ranks} reveals that while SVD centrality yields plausible rankings—correlating broadly with PageRank—it captures a fundamentally different physical property of the network.

\begin{figure}[t]
\centering
\includegraphics[width=\textwidth]{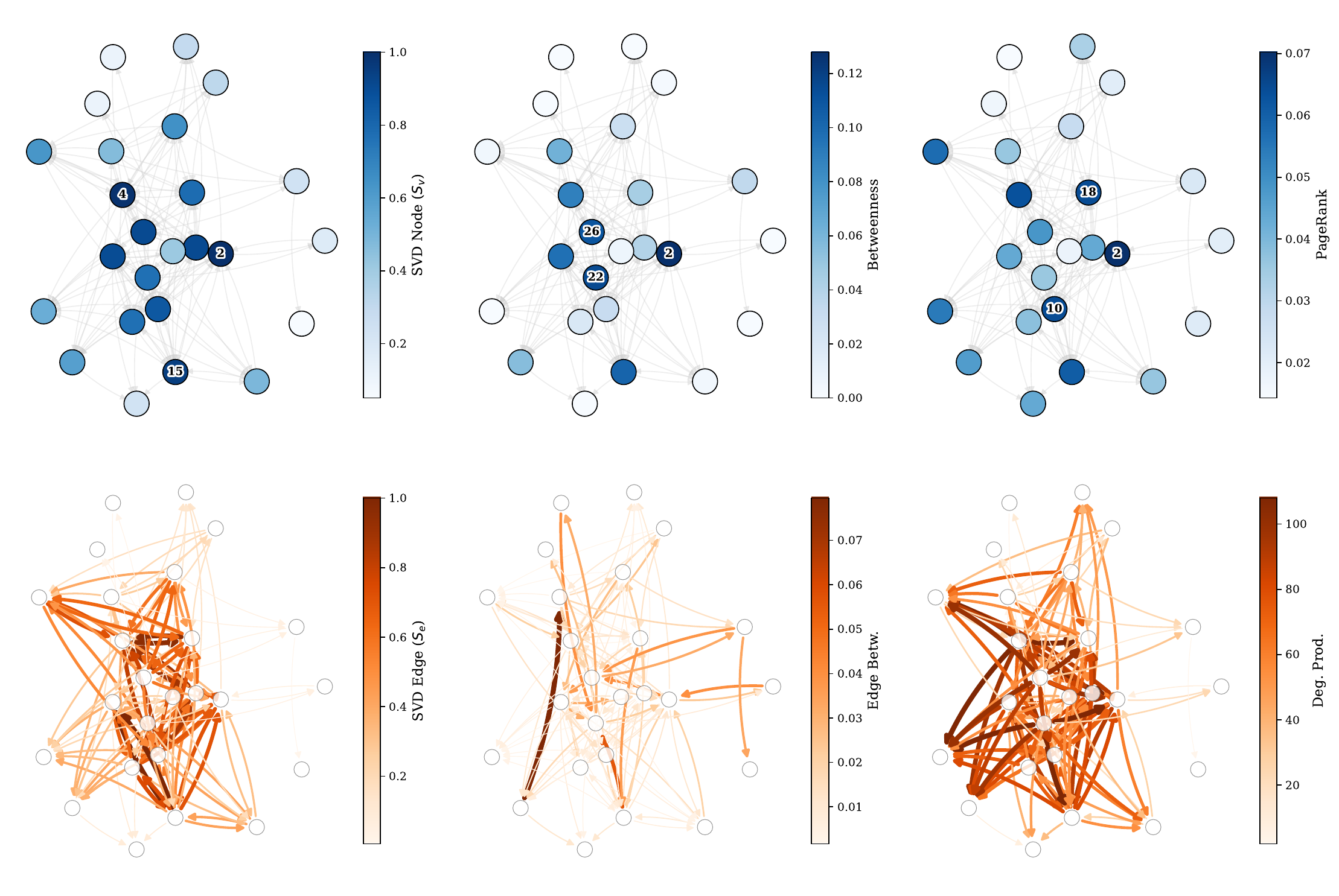}
\caption{General centrality analysis of the Dutch school friendship network. \textbf{Row 1 (Node Centrality)}: SVD ($S_v$) identifies Students 2 and 4 as global integrators, showing broad agreement with PageRank but distinct deviations from Betweenness. \textbf{Row 2 (Edge Centrality)}: The contrast is striking; SVD ($S_e$) reveals a dense ``diffusive backbone'' of ties, whereas Edge Betweenness (bottom right) is extremely sparse, highlighting only critical bridges. This dense ranking confirms SVD captures the capacity for broad information diffusion rather than just geodesic transport.}
\label{fig:dutch_general}
\end{figure}

The divergence between SVD and Betweenness Centrality highlights the distinction between \textit{global integration} and \textit{local bridging}. As shown in Table~\ref{tab:dutch_general_ranks}, Student 22 ranks 2nd in Betweenness but does not appear in the top 5 for SVD or PageRank. Physically, this suggests Student 22 acts as a bottleneck or bridge between clusters—critical for shortest-path routing but less central to the network's diffusive steady state. Conversely, Student 2 ($S_v=1.00$) dominates the spectral rankings (SVD and PageRank) but scores much lower in Betweenness (0.128). Although these scores are normalized independently and represent different structural properties, the substantial relative difference identifies Student 2 as the ``core'' of the social graph—the node most accessible via random walks or diffusive processes, even if they do not control the exclusive shortest paths. SVD centrality thus provides a metric for \textit{diffusive influence}, which is arguably more relevant for social contagion than geodesic efficiency.

\begin{table}[t]
\small
\setlength{\tabcolsep}{4pt}
\centering
\caption{Top 5 students by general node centrality (Wave 3). The ranking divergence reveals complementary structural roles: Student 22 (high Betweenness, low SVD) acts as a local bridge, while Student 2 (top SVD/PageRank) serves as the global integrator. This suggests that SVD centrality highlights a network's diffusive core—nodes that are most accessible under diffusive (spectral) dynamics—complementary to shortest-path-based notions of centrality.}
\label{tab:dutch_general_ranks}
\begin{tabular}{lccc}
\toprule
\textbf{Rank} & \textbf{SVD Node ($S_v$)} & \textbf{Betweenness} & \textbf{PageRank} \\
\midrule
1 & Stud.\ 2 (1.000) & Stud.\ 2 (0.128) & Stud.\ 2 (0.070) \\
2 & Stud.\ 4 (0.992) & Stud.\ 22 (0.115) & Stud.\ 10 (0.064) \\
3 & Stud.\ 15 (0.946) & Stud.\ 26 (0.109) & Stud.\ 18 (0.064) \\
4 & Stud.\ 7 (0.905) & Stud.\ 15 (0.102) & Stud.\ 4 (0.063) \\
5 & Stud.\ 26 (0.905) & Stud.\ 6 (0.096) & Stud.\ 15 (0.060) \\
\bottomrule
\end{tabular}
\end{table}

The directional analysis (\Cref{fig:dutch_directional}, Table~\ref{tab:dutch_directional_ranks}) demonstrates how the SVD framework decomposes influence into orthogonal components without the instability issues often associated with heuristic iterations.

\begin{table}[h!]
\small
\setlength{\tabcolsep}{4pt}
\centering
\caption{Top 5 students by directional roles. SVD centrality unifies the Hub and Authority roles in Student 4, identifying the network's primary feedback loop. In contrast, HITS splits these roles, potentially over-segmenting the dense social core. This demonstrates SVD's ability to identify robust, dual-role influencers.}
\label{tab:dutch_directional_ranks}
\begin{tabular}{lcccc}
\toprule
\textbf{Rank} & \textbf{SVD Auth ($S_{auth}$)} & \textbf{HITS Auth} & \textbf{SVD Hub ($S_{hub}$)} & \textbf{HITS Hub} \\
\midrule
1 & Stud.\ 4 (1.000) & Stud.\ 3 (0.073) & Stud.\ 4 (1.000) & Stud.\ 4 (0.096) \\
2 & Stud.\ 10 (0.977) & Stud.\ 8 (0.073) & Stud.\ 7 (0.863) & Stud.\ 26 (0.091) \\
3 & Stud.\ 15 (0.932) & Stud.\ 23 (0.069) & Stud.\ 26 (0.751) & Stud.\ 7 (0.085) \\
4 & Stud.\ 18 (0.926) & Stud.\ 2 (0.068) & Stud.\ 10 (0.737) & Stud.\ 2 (0.077) \\
5 & Stud.\ 7 (0.869) & Stud.\ 18 (0.065) & Stud.\ 15 (0.705) & Stud.\ 10 (0.076) \\
\bottomrule
\end{tabular}
\end{table}

\begin{figure}[t]
\centering
\includegraphics[width=\textwidth]{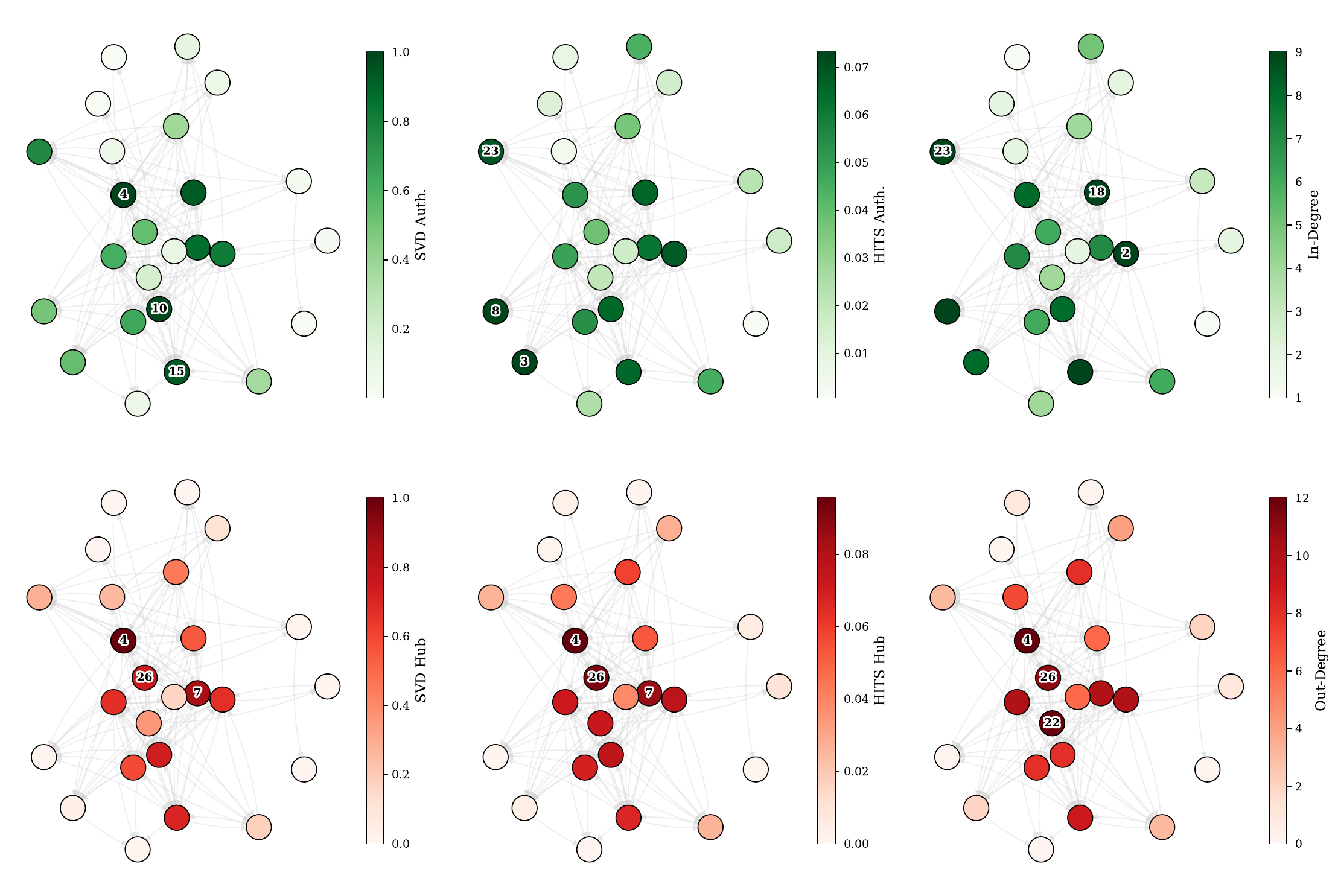}
\caption{Directional role analysis. \textbf{Row 1 (Authority)}: SVD Authority ($S_{auth}$) maps the ``potential'' landscape, identifying students with high social capital (incoming ties). \textbf{Row 2 (Hub)}: SVD Hub ($S_{hub}$) maps the ``flow'' landscape (outgoing ties). Student 4 is identified as the dual-role leader, maximizing both metrics.}
\label{fig:dutch_directional}
\end{figure}

A key insight emerges from the comparison with HITS. While HITS identifies different leaders for Authority (Student 3) and Hub (Student 4), our spectral framework identifies Student 4 as the dominant actor in \textit{both} dimensions ($S_{auth}=1.00, S_{hub}=1.00$). This suggests that Student 4 represents the ``spectral backbone'' of the class—a super-spreader who is both highly prestigious and actively engaged. SVD's identification of this coupled role is physically significant: in dense, reciprocal social groups, influence and activity are often reinforced by the same feedback loops. By diagonalizing the incidence matrix directly, our method captures this fundamental mode of the social field, offering a robust ``new tool'' for researchers to detect key players who drive the system's dynamics through both reception and transmission.

\textbf{Note on Student 2.}
As defined above (see Sec.~3 and Sec.~4.7), $S_v$ and $S_e$ originate respectively from the vertex and edge
Laplacians ($L_0$ and $L_1$). Hence $S_v$ measures global accessibility while $S_{hub},S_{auth}$ are formed by
aggregating edge-level scores $S_e$. Consequently, a node can be highly accessible (high $S_v$) even though its
incident edges are not spectrally dominant (moderate $S_e$), yielding moderate directional aggregates.
Student 2 therefore exemplifies a node that is centrally accessible in a diffusive sense but not specialized as a
strong broadcaster or receiver. This clarifies the apparent discrepancy between the global ranking in
Table~\ref{tab:dutch_general_ranks} and the directional rankings in Table~\ref{tab:dutch_directional_ranks}.

\subsubsection{Biological Systems}

Metabolic and protein interaction networks represent fundamental biological systems where directed biochemical reactions and regulatory interactions define cellular function \citep{Gao2014}. We analyze two prominent directed networks: the metabolic network of \emph{C. elegans} (nematode worm) and the protein-protein interaction (PPI) network of \emph{S. cerevisiae} (yeast), providing validation in molecular biology contexts where spectral properties reveal essential system components.

The \emph{C. elegans} metabolic network (453 metabolites, 2,025 reactions) displays a branched structure radiating from a central dense core, as shown in \Cref{fig:c_elegans_ppi} \citep{duch2005community}. The top row demonstrates how hub-authority decomposition reveals directional organization. SVD hub centrality (left panel, red) concentrates on a cluster of nodes in the central region, identifying metabolites that initiate multiple downstream pathways. SVD authority centrality (center panel, blue) highlights a partially overlapping but spatially distinct set of central nodes, corresponding to metabolites where multiple pathways converge. The spatial separation between hubs and authorities within the core reveals an underlying directional flow structure in the metabolic network. Betweenness centrality (right panel, green) identifies a different structural property, highlighting nodes that lie on many shortest paths, though it assigns low values to most peripheral nodes.

The bottom row examines reactions rather than metabolites. SVD edge centrality (left panel, purple) produces a continuous gradient throughout the network, with high-centrality reactions in the dense core and meaningful variation extending into peripheral branches. This continuous distribution enables discrimination among reactions at different structural positions. In-degree centrality (center panel, orange), here aggregated from vertices to their incident edges to provide a baseline for local connectivity, reveals one particularly high-degree node in the core—likely a metabolite participating in many transformations. Betweenness edge centrality (right panel, brown) shows a different pattern, with most reactions receiving very low values and only a few edges in the dense core region showing visible coloration. The sparsity of this distribution reflects the concentration of shortest paths through a limited subset of reactions.

\begin{figure}[t]
\centering
\includegraphics[width=\textwidth]{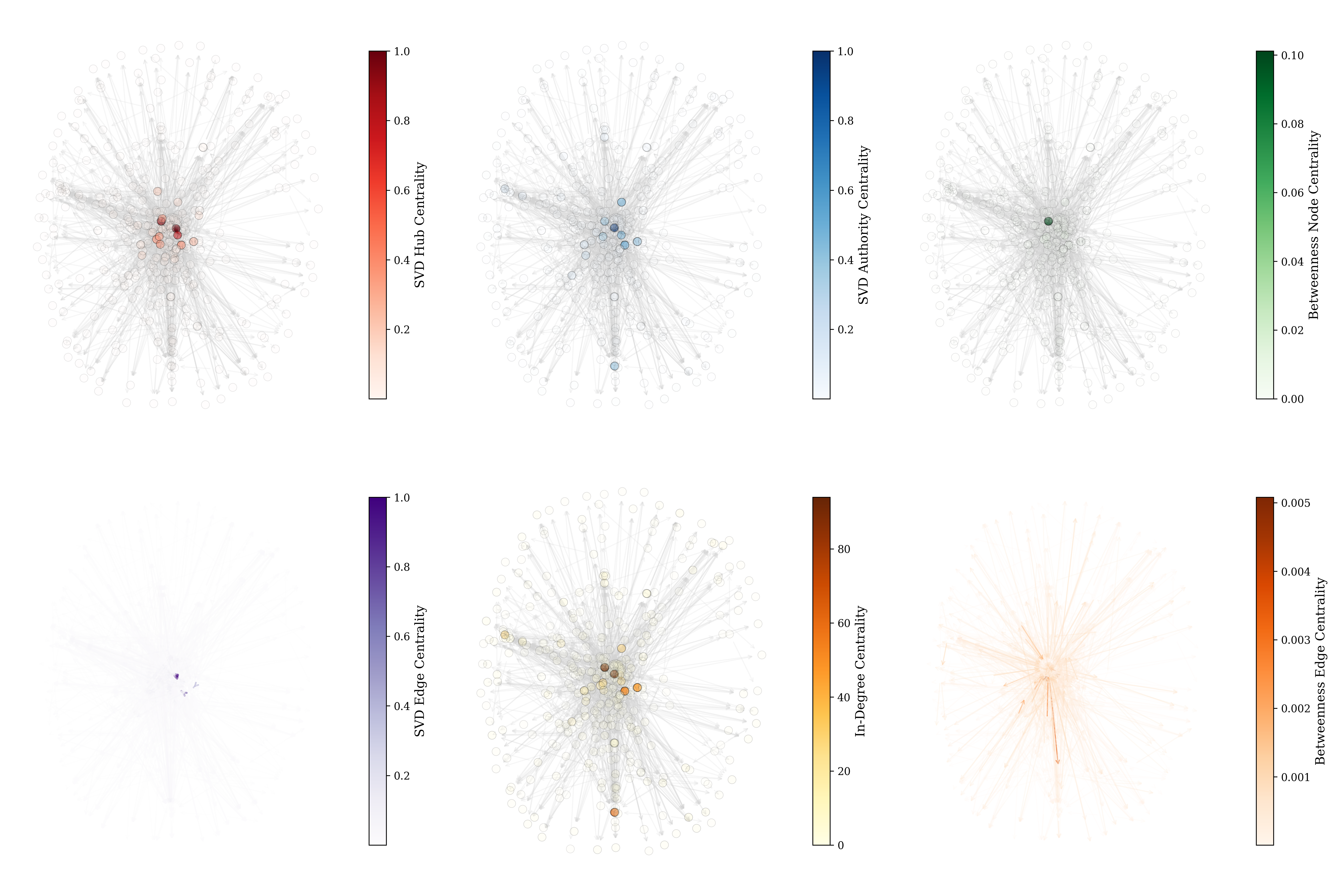}
\caption{SVD centrality analysis of the \emph{C. elegans} metabolic network (453 metabolites, 2,025 reactions) \citep{duch2005community}. Top row: Hub, Authority, and Betweenness node centralities. Bottom row: SVD Edge, In-Degree, and Betweenness edge centralities. Hub and authority measures identify spatially distinct subsets of core metabolites, revealing directional flow structure. Specifically, while certain nodes like Node 2 exhibit high SVD node centrality due to their global structural role as integrators, they may not dominate the directional hub or authority scores, which isolate pure broadcasting and receiving behaviors.}
\label{fig:c_elegans_ppi}
\end{figure}

The yeast PPI network scales to 1,870 proteins and 4,480 interactions, exhibiting a densely packed circular structure with concentric layers (\Cref{fig:yeast_ppi}) \citep{coulomb2005gene}. This layout reflects modular organization where highly connected proteins cluster toward the center. SVD hub centrality (top left, red) identifies proteins in the inner layers with high outgoing interaction propensity, while SVD authority centrality (top center, blue) highlights proteins that receive many incoming interactions. The substantial overlap between high-hub and high-authority proteins in the inner core suggests many proteins serve dual regulatory roles. Betweenness centrality (top right, green) shows only faint coloration on a handful of innermost proteins, with most nodes displaying near-zero values. This suggests that shortest paths in this dense network are distributed across many alternative routes rather than concentrating through specific bottlenecks.

The edge analysis reveals distinct patterns across measures. SVD edge centrality (bottom left, purple) identifies a ring-like pattern of high-centrality interactions distributed throughout inner and middle layers. This ring structure represents a backbone maintaining network cohesion across functional modules. The pattern shows density variations suggesting different levels of module connectivity in different network regions. In-degree centrality (bottom center, orange) highlights individual proteins receiving many interactions, with several inner core nodes showing very high values. The distribution decays gradually from core to periphery, reflecting hierarchical organization. Betweenness edge centrality (bottom right, brown) appears almost entirely uniform at very low values, with minimal coloration visible only on a tiny fraction of edges in the network center. This extreme sparsity indicates that shortest paths are distributed across numerous alternative routes, a signature of redundancy in biological networks.

\begin{figure}[t]
\centering
\includegraphics[width=\textwidth]{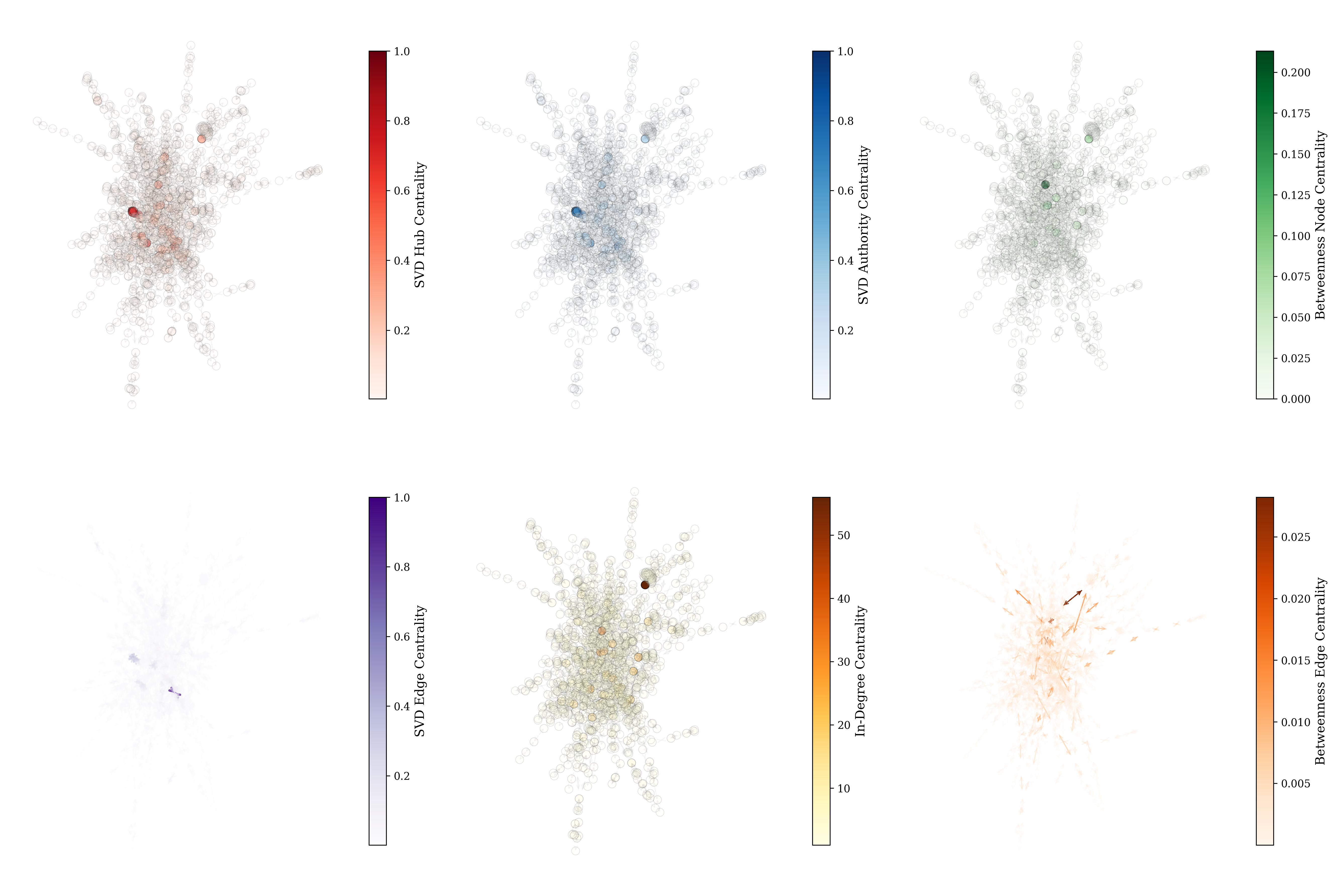}
\caption{SVD incidence centrality analysis of the yeast protein-protein interaction network \citep{coulomb2005gene}. SVD edge centrality reveals a ring-like backbone of structurally critical interactions, while betweenness edge centrality assigns near-zero values to most interactions due to high path redundancy.}
\label{fig:yeast_ppi}
\end{figure}

\subsubsection{Transportation Systems}

Transportation networks provide concrete physical interpretation for spectral centrality, as they represent directed flow systems where our measures correspond to concepts like traffic distribution and network accessibility.

The OpenFlights worldwide airline network (2,905 airports, 30,442 routes) shows a dense star-like structure with extensive connectivity radiating from a central cluster (\Cref{fig:openflights}) \citep{openflights2024}. Due to its large scale, we employed a truncated SVD approach as outlined in \Cref{par:truncated_svd}, approximating the full SVD by retaining only the top-$k$ singular values and vectors. We used binary search to determine the largest truncation rank $k$ that remained computationally feasible within our runtime and memory constraints—a practical necessity given the matrix dimensions of $2{,}905 \times 30{,}442$. As an alternative or complementary criterion, one could also select $k$ such that a significant portion of the spectral energy (measured by $\sum_{i=1}^k \sigma_i^2 / \sum_{i=1}^r \sigma_i^2$) is preserved, ensuring that the dominant structural patterns governing global accessibility are captured. This reflects the hub-and-spoke architecture of modern aviation. SVD hub centrality (top left, red) identifies multiple airports in the central cluster as major origin points with extensive outgoing connectivity. SVD authority centrality (top center, blue) highlights a partially overlapping set as primary destination aggregators. The overlap indicates that major hubs serve dual roles as both generators and receivers of traffic, though subtle differences reveal asymmetric flow patterns. Betweenness centrality (top right, green) shows moderate values for a subset of core airports but near-zero values for most peripheral airports, indicating they do not serve as critical transfer points in shortest-path routing.

For edges, SVD centrality (bottom left, purple) produces a radial gradient emanating from the core, with high-centrality routes connecting hubs and meaningful values extending along radiating arms. This captures both major trunk routes and regionally critical connections. In-degree centrality (bottom center, orange) highlights airports receiving many incoming routes, showing heterogeneity even among major hubs. Betweenness edge centrality (bottom right, brown) appears almost invisible, with only a handful of edges between central hubs showing minimal coloration. The extreme sparsity reflects high redundancy in global aviation where multiple routing options exist between most airport pairs.

\begin{figure}[t]
\centering
\includegraphics[width=\textwidth]{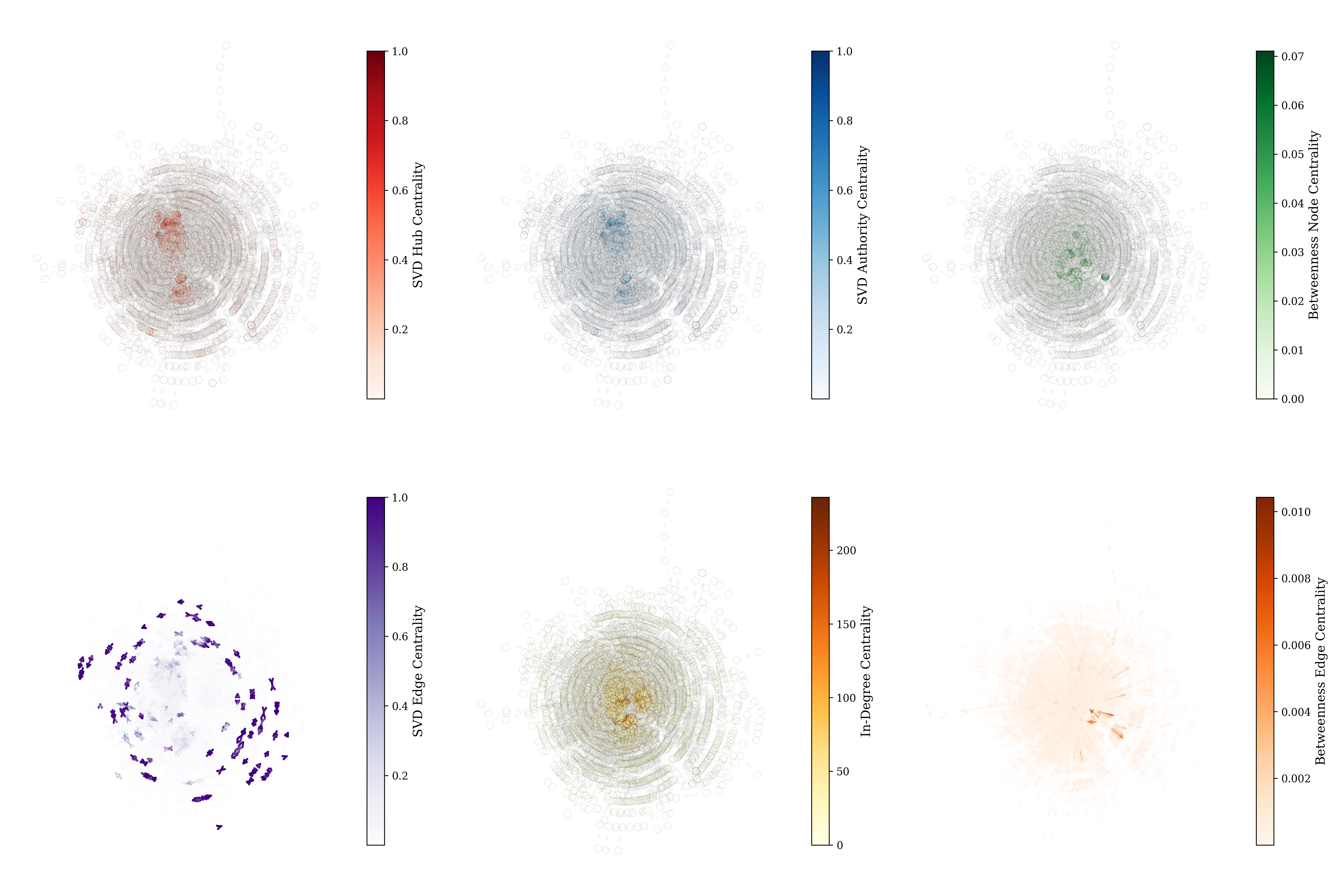}
\caption{SVD incidence centrality analysis of the OpenFlights global airline network \citep{openflights2024} (2,905 airports, 30,442 flight routes). The spectral framework produces dense rankings identifying both trunk routes and regionally critical connections.}
\label{fig:openflights}
\end{figure}

The European road network (1,174 cities, 1,417 roads) exhibits a radial pattern reflecting geographic distribution and hierarchical infrastructure organization (\Cref{fig:euroroad}) \citep{subelj2011robust}; the visualization layout is based on these geographic city coordinates. SVD hub centrality (top left, red) identifies cities in the central-lower region as major traffic generation points, while SVD authority centrality (top center, blue) highlights partially overlapping cities as major destination points. The spatial differences indicate directional asymmetries in traffic flow. Betweenness centrality (top right, green) shows sparse coloration on a handful of central cities, with most nodes displaying near-zero values.

SVD edge centrality (bottom left, purple) reveals clear hierarchical backbone structure: a prominent trunk of high-centrality roads—appearing as a dense, dark purple central axis—running through the network center, appearing as a potential-source component that organizes the broader connectivity, with additional arterial branches extending to multiple regions. Roads with intermediate values form a denser network of secondary connections throughout the central region. This hierarchical pattern distinguishes primary continental arteries from secondary regional connections from local routes, all based on global spectral properties rather than local topology. In-degree centrality (bottom center, orange) shows nodes with higher values concentrated in the central area, identifying natural junction points. Betweenness shows a similar pattern to the other measures.

\begin{figure}[t]
\centering
\includegraphics[width=\textwidth]{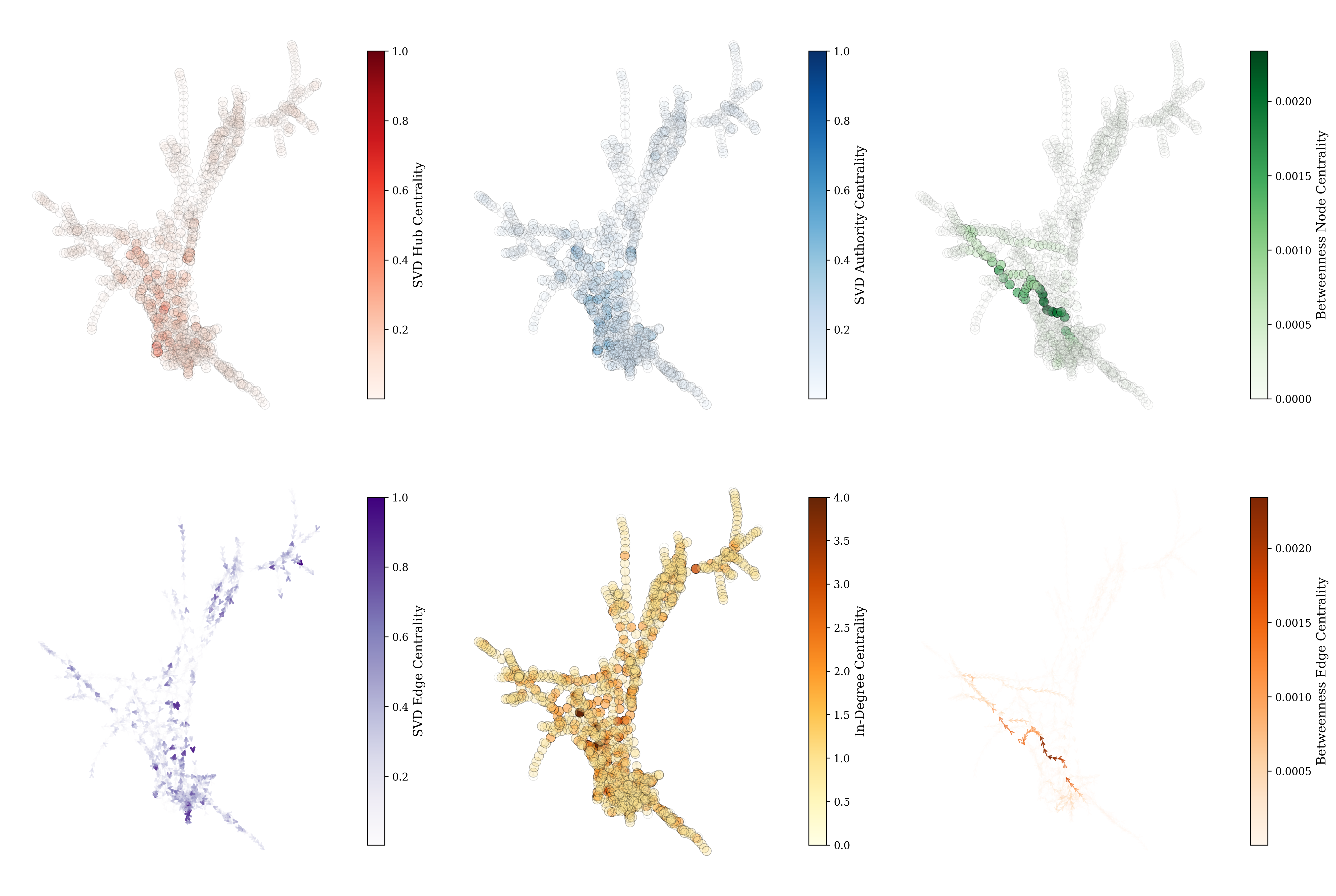}
\caption{SVD incidence centrality analysis of the European road network \citep{subelj2011robust} (1,174 cities, 1,417 roads). SVD edge centrality identifies hierarchical backbone structure with primary arteries and secondary connections based on global spectral properties.}
\label{fig:euroroad}
\end{figure}

Across these networks, a consistent pattern emerges in the relationship between spectral and path-based measures. SVD centrality produces continuous distributions that discriminate across most network elements, while betweenness concentrates on small subsets. This difference reflects their computational foundations: SVD leverages global spectral properties of the incidence matrix, while betweenness enumerates shortest paths. In dense or highly connected networks with abundant alternative paths, betweenness assigns similar low values to most elements, whereas spectral measures maintain discrimination through the eigenstructure of the Hodge Laplacians. This distinction has practical implications for applications requiring ranked prioritization across network elements.

\subsection{Hypergraph Networks}
\label{sec:hypergraph_experiments}

After validating our framework on directed graphs, we examine its behaviour on hypergraphs, where hyperedges connect multiple vertices simultaneously. We evaluate four real-world systems drawn from the XGI dataset library\footnote{Datasets from XGI library: \url{https://xgi.readthedocs.io/en/stable/xgi-data.html}.}: an ecological plant–pollinator network \citep{Bek2006PollinationNetwork}, human disease associations \citep{goh2007human}, hospital contact patterns \citep{vanhems2013estimating}, and senate bill co-sponsorship \citep{landry2024senate-bills}. \Cref{fig:hypergraph_comparison} presents a side-by-side visual comparison that highlights how SVD incidence centrality and hypergraph betweenness emphasize different structural features across domains.

We adopt a unified visual convention: node and hyperedge scores are displayed on a continuous colormap and, where necessary for cross-panel readability, scores are linearly rescaled for display only. SVD-derived centralities are computed directly from the hypergraph incidence matrix, which yields deterministic, reproducible scores. By contrast, betweenness for hypergraphs requires an explicit path model (treating a hyperedge as a single traversal step, or as an internal clique, etc.) and optional normalization; these modelling choices affect both absolute values and spatial patterns. We therefore interpret betweenness visualizations as implementation- and normalization-dependent, and we focus on qualitative, comparative statements that are robust to reasonable preprocessing choices.

\Cref{fig:hypergraph_comparison} (first row) shows the pollinator network, a small bipartite system with 16 labeled nodes. SVD node centrality produces a clear gradient that distinguishes species with high spectral influence: darker nodes (indicating higher centrality in our colormap) at the structural boundaries correspond to species that participate in low-frequency, system-spanning modes and thus act as structural generalists, while lighter nodes correspond to more specialized species. SVD hyperedge centrality likewise discriminates interaction groups according to their contribution to collective modes. Betweenness in this small hypergraph appears symmetric and clustered near a central value in the chosen visualization scale; this pattern most likely reflects the combined effect of (i) the small network size, which limits the range of distinct shortest hyperpaths, and (ii) the particular normalization or centering applied for display. The practical implication is that for very small bipartite hypergraphs, path-based measures may lack the resolution to distinguish graded structural roles that spectral measures readily reveal.

\begin{figure}[!htbp]
\centering


\begin{tabular}{@{}cccc@{\hspace{0.4cm}}c@{}}
\toprule
\multicolumn{2}{c}{\textbf{Nodes}} & \multicolumn{2}{c}{\textbf{Hyperedges}} & \\
\cmidrule(lr){1-2} \cmidrule(lr){3-4}
\textbf{SVD} & \textbf{Betweenness} & \textbf{SVD} & \textbf{Betweenness} & \\
\midrule
\includegraphics[width=0.22\textwidth]{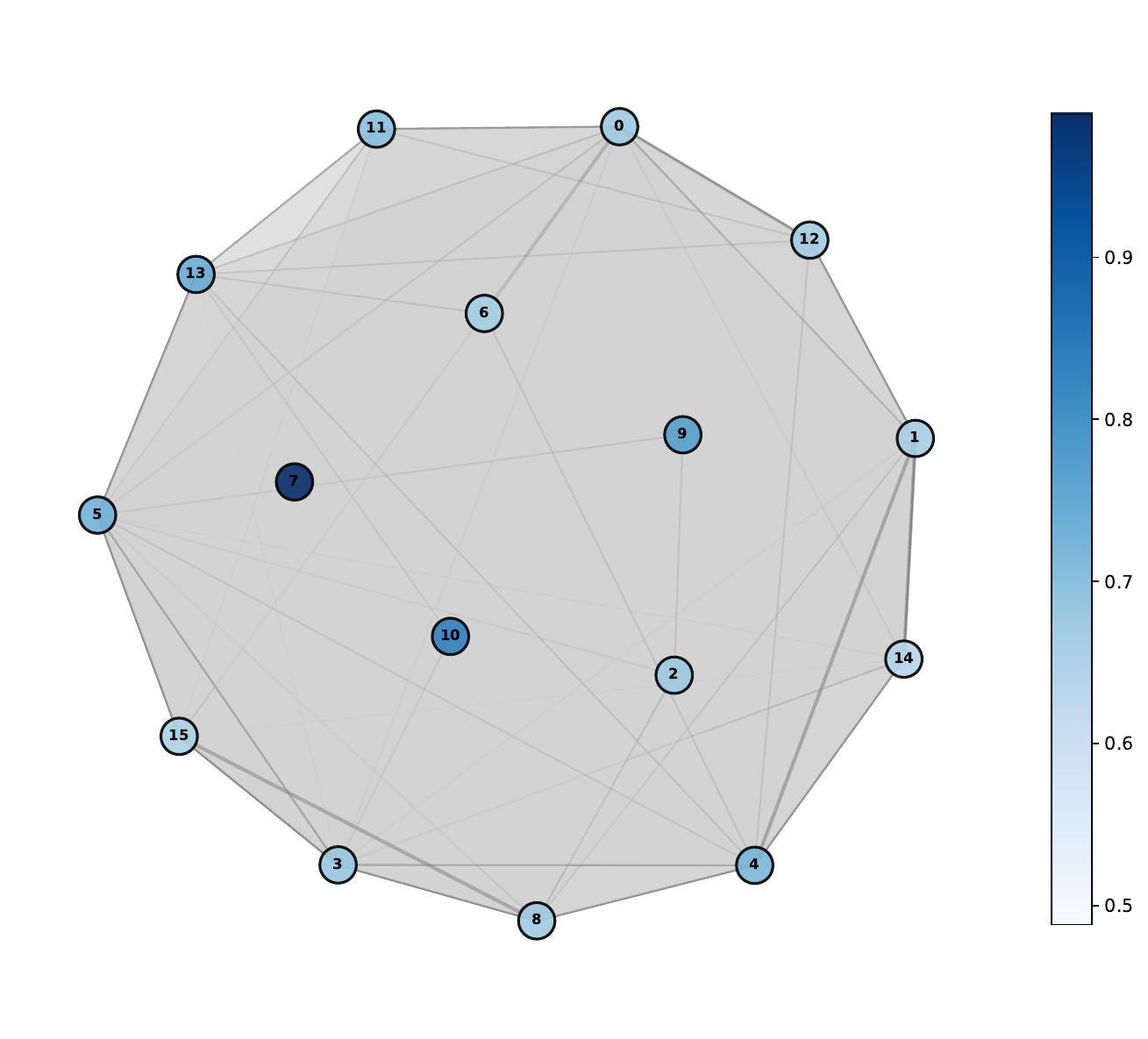} &
\includegraphics[width=0.22\textwidth]{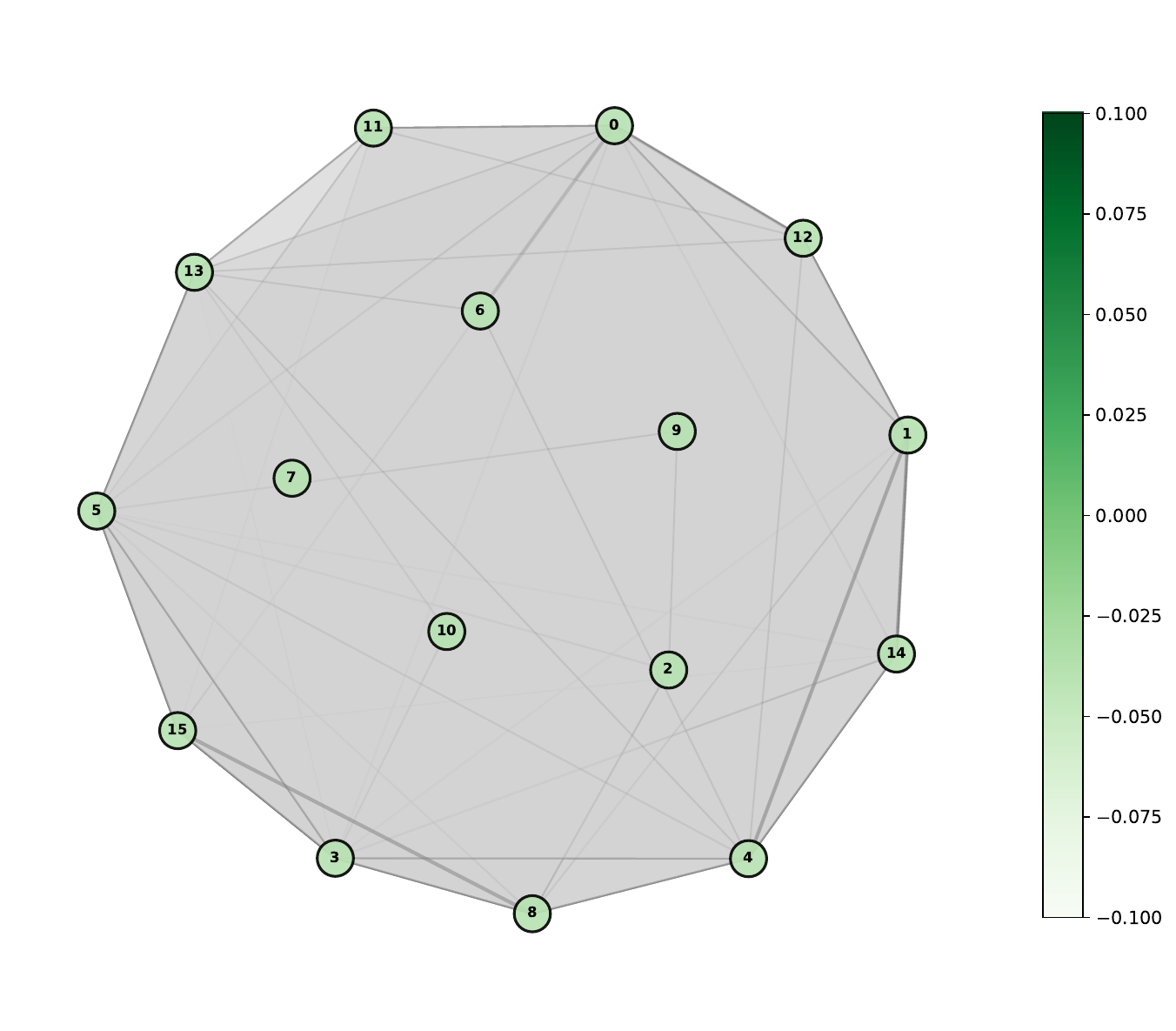} &
\includegraphics[width=0.22\textwidth]{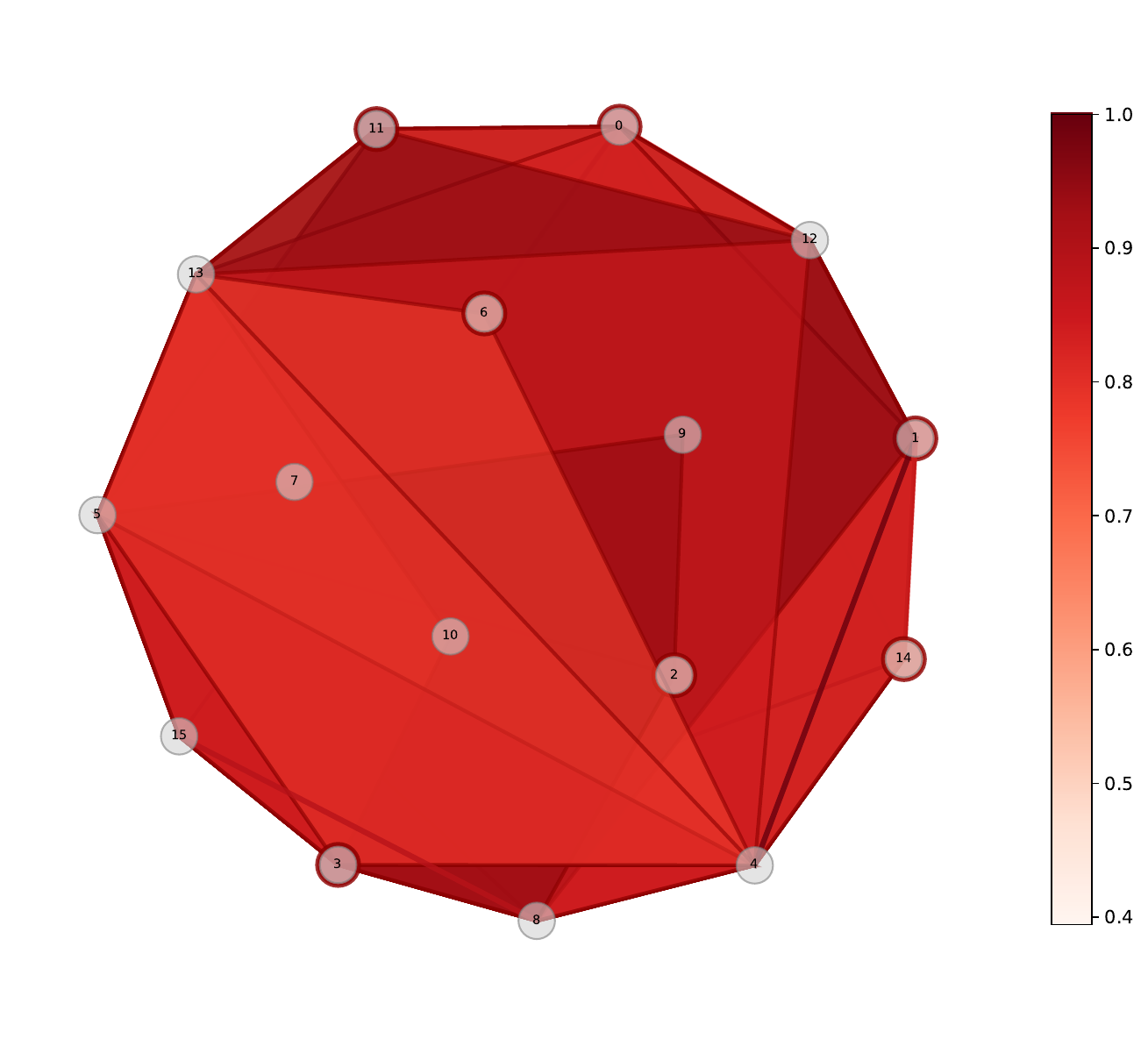} &
\includegraphics[width=0.22\textwidth]{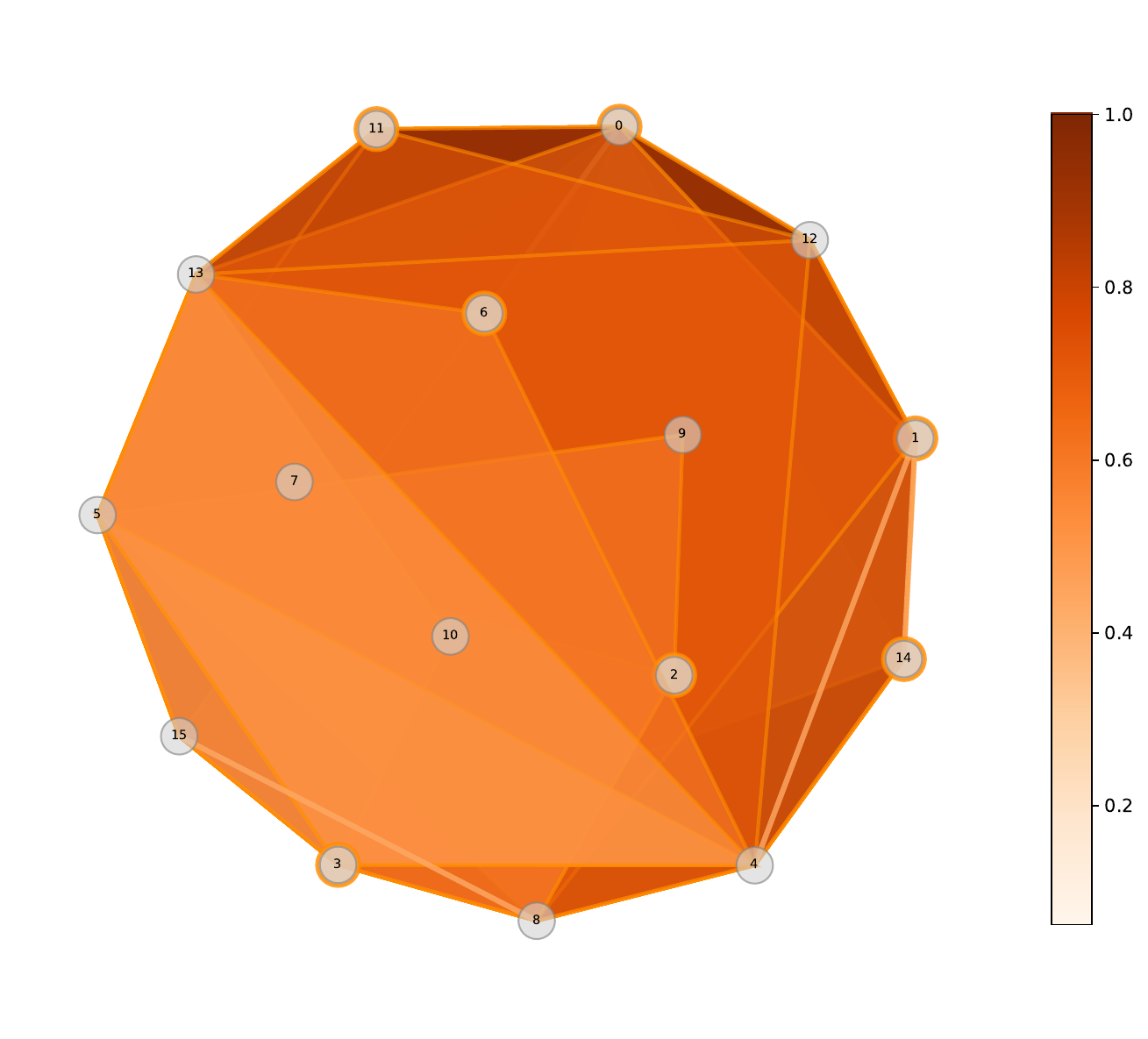} &
\raisebox{0.08\textwidth}{\rotatebox{90}{\textbf{Pollinator}}} \\[0.3cm]
\noalign{\hrule height 0.4pt}

\includegraphics[width=0.22\textwidth]{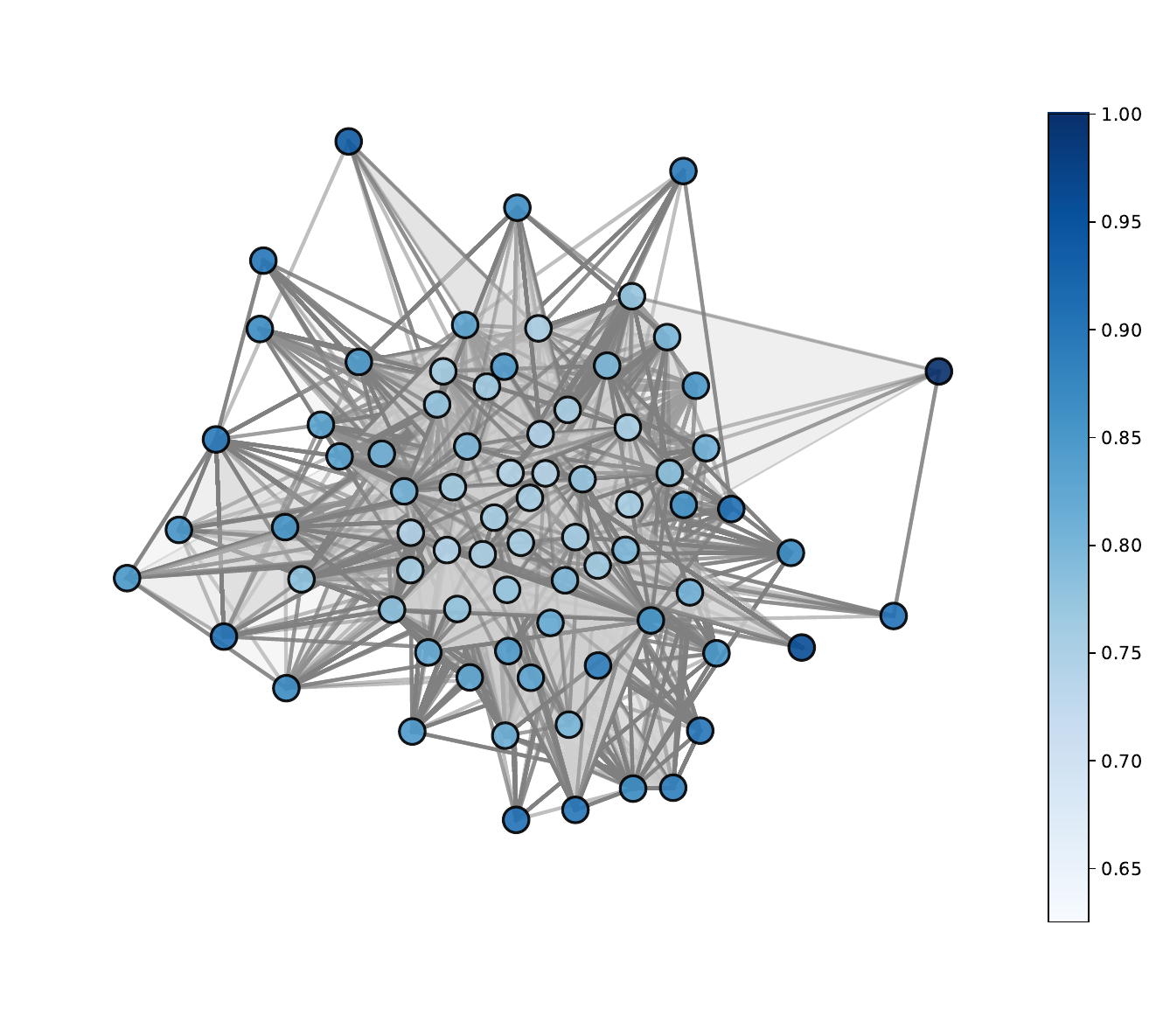} &
\includegraphics[width=0.22\textwidth]{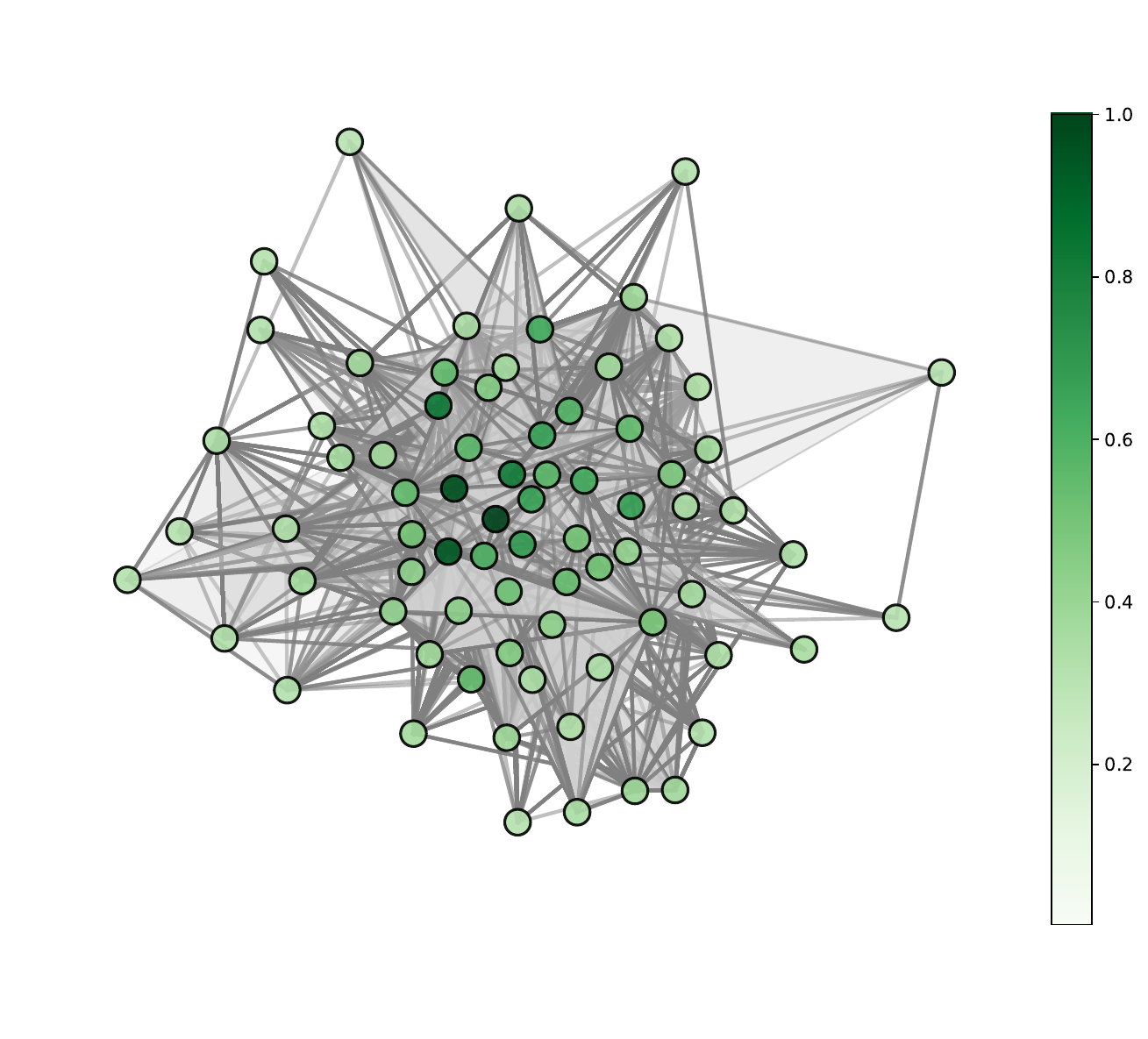} &
\includegraphics[width=0.22\textwidth]{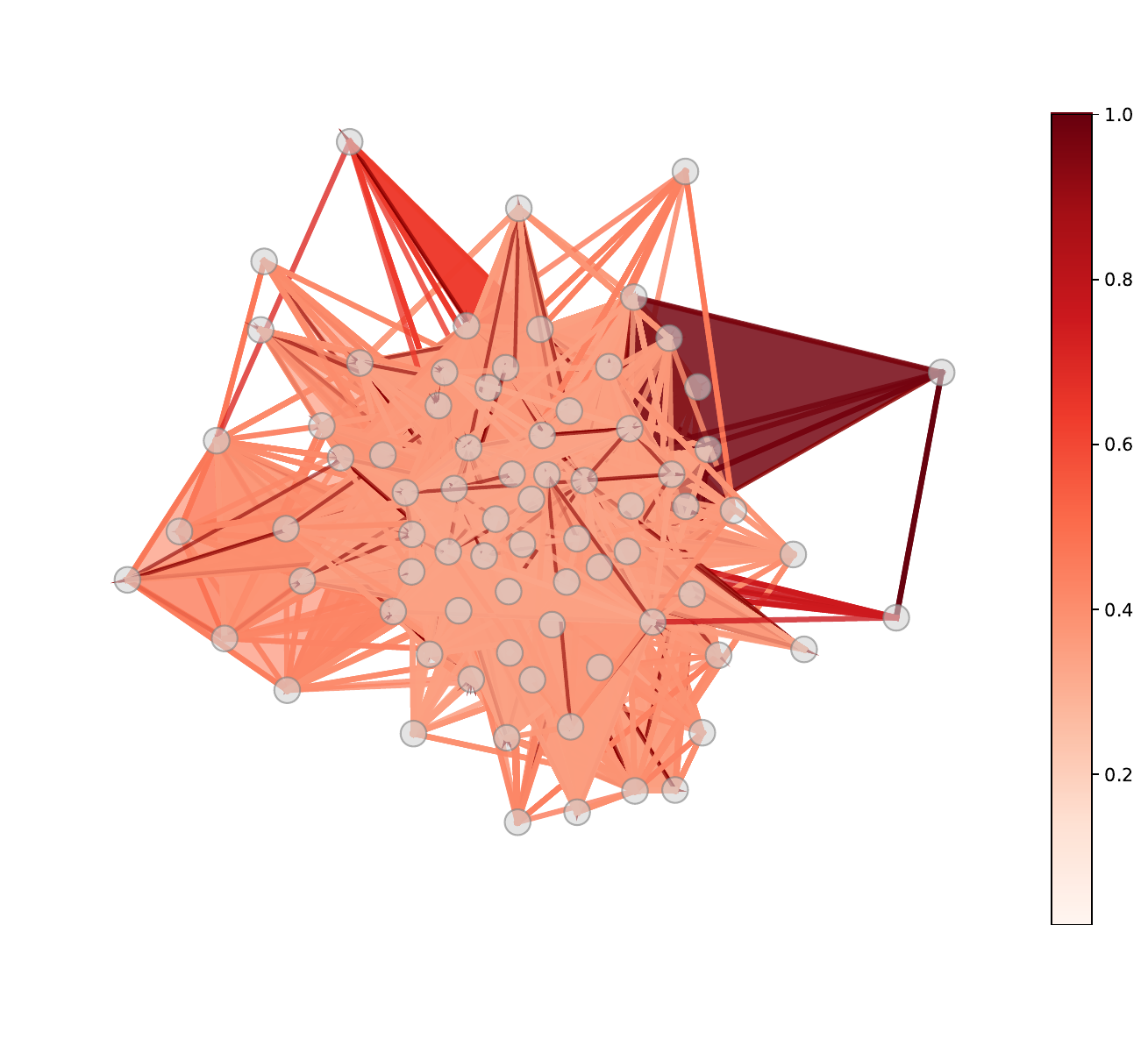} &
\includegraphics[width=0.22\textwidth]{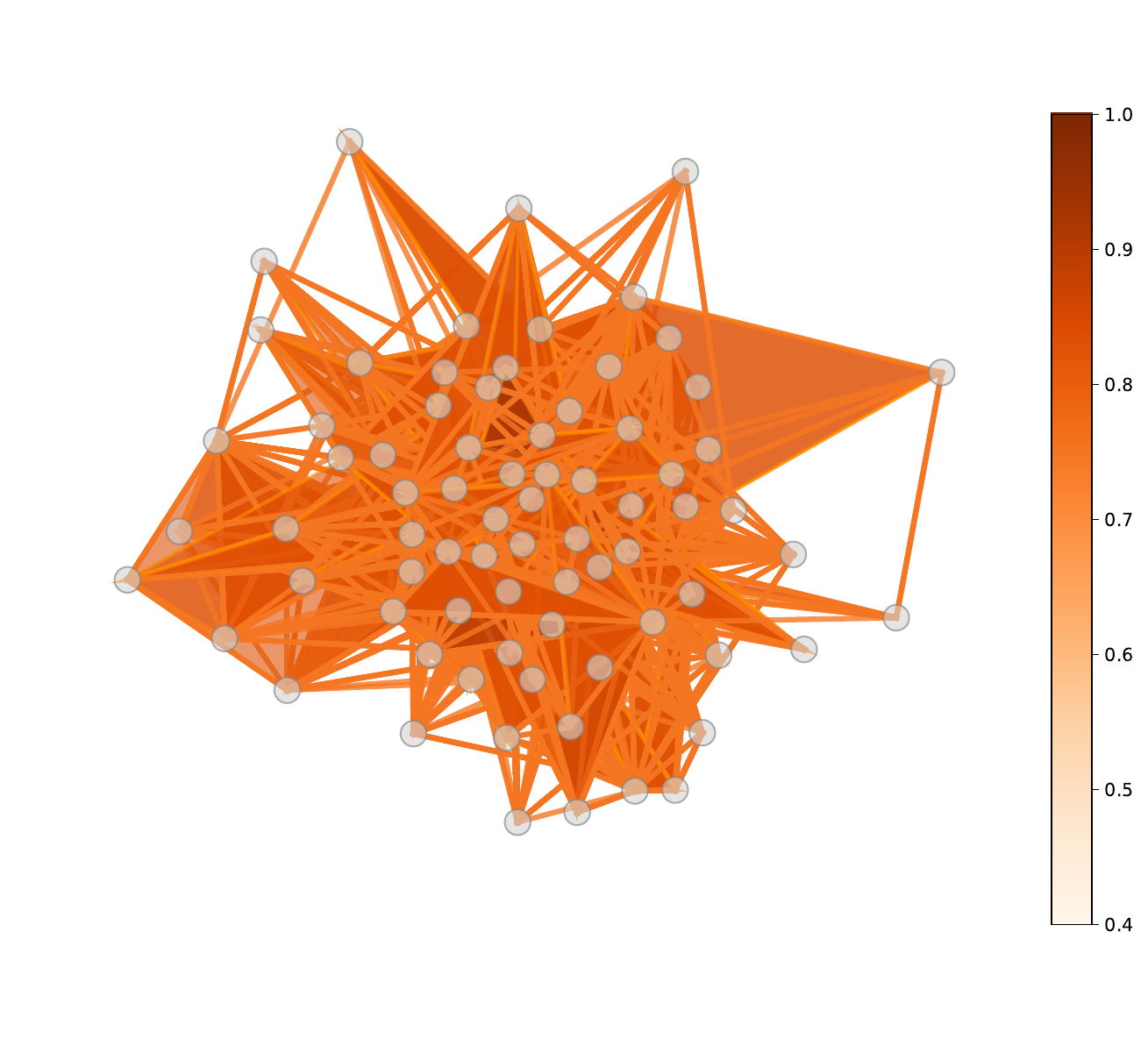} &
\raisebox{0.08\textwidth}{\rotatebox{90}{\textbf{Hospital}}} \\[0.3cm]
\noalign{\hrule height 0.4pt}

\includegraphics[width=0.22\textwidth]{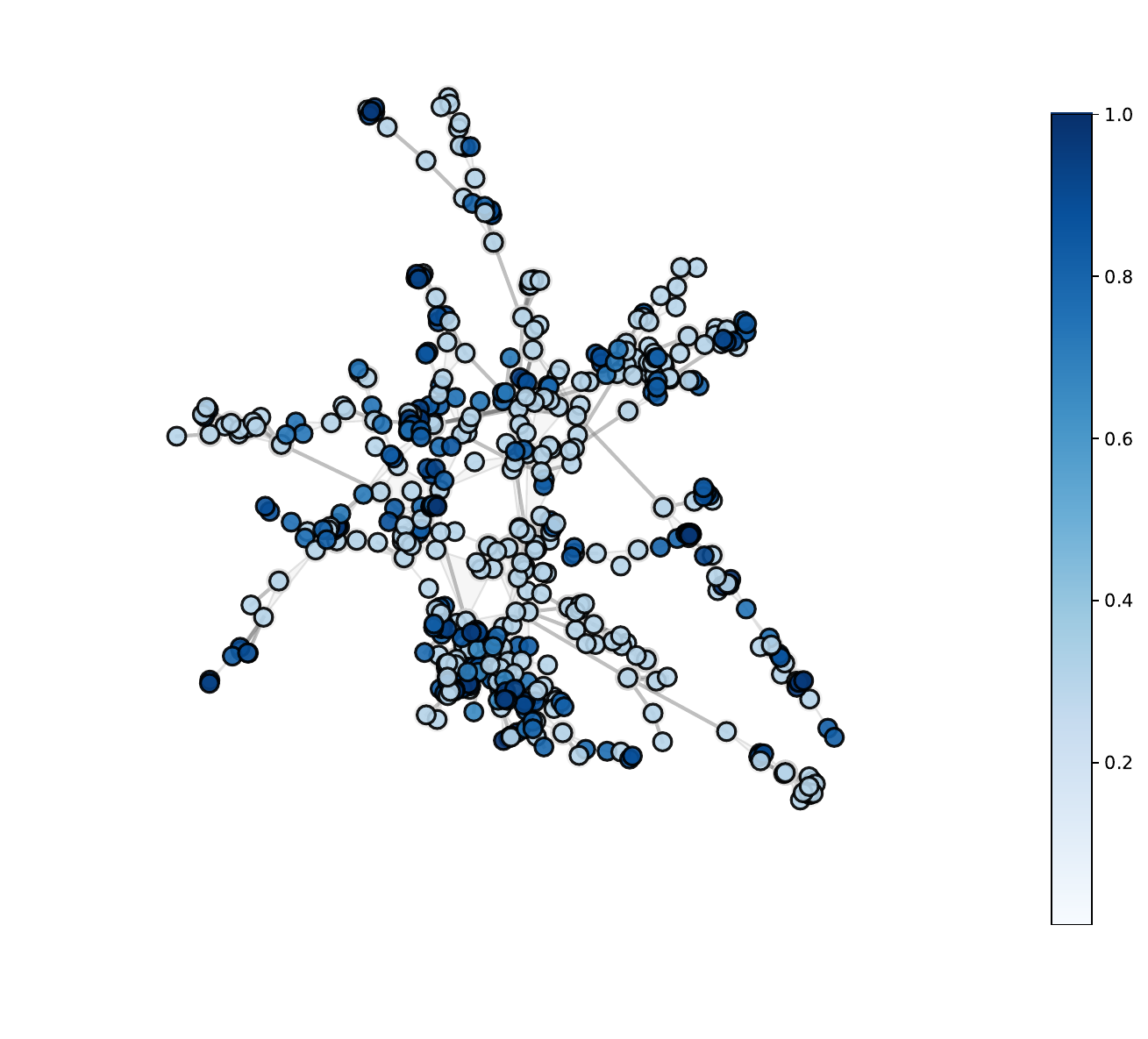} &
\includegraphics[width=0.22\textwidth]{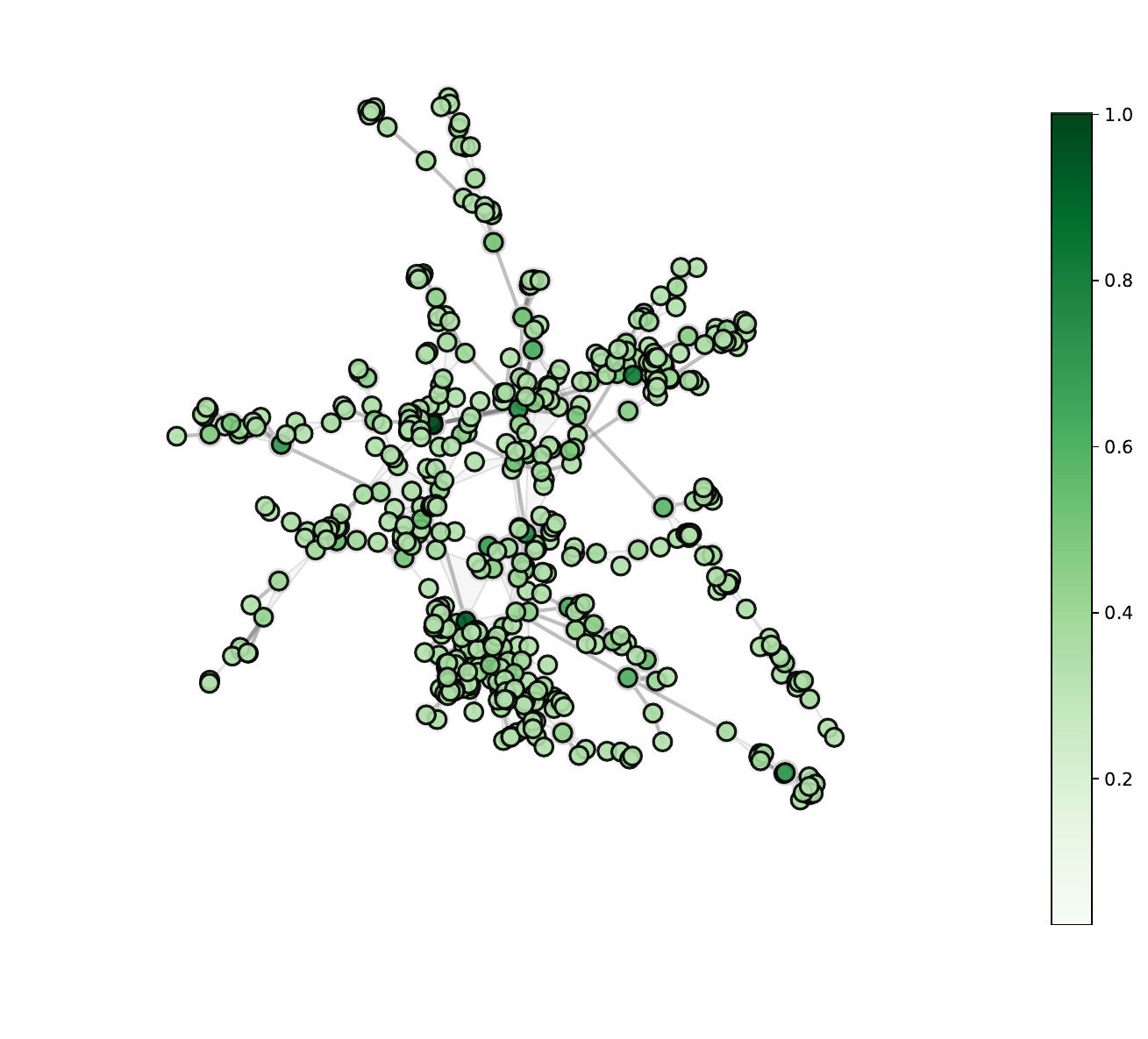} &
\includegraphics[width=0.22\textwidth]{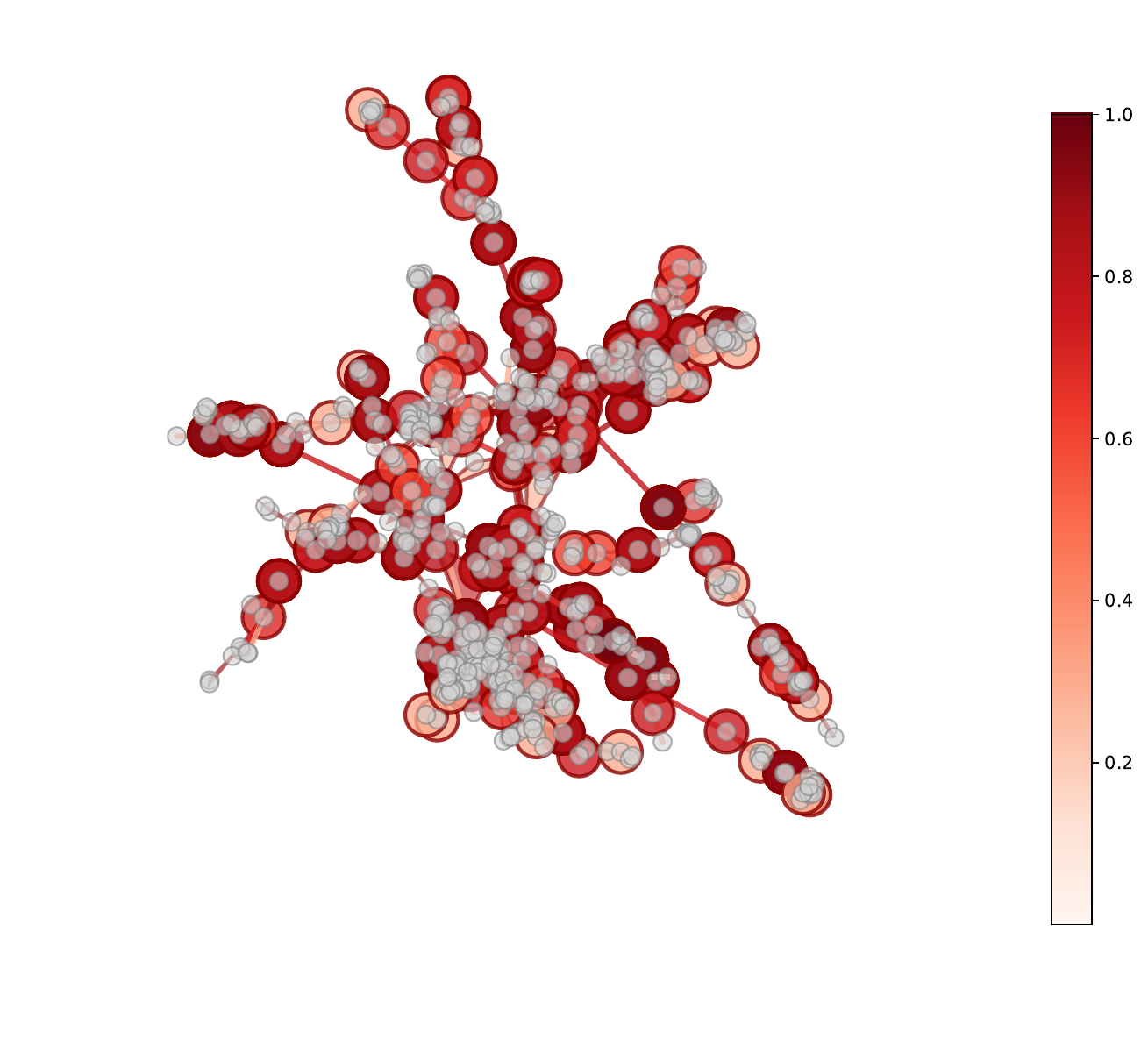} &
\includegraphics[width=0.22\textwidth]{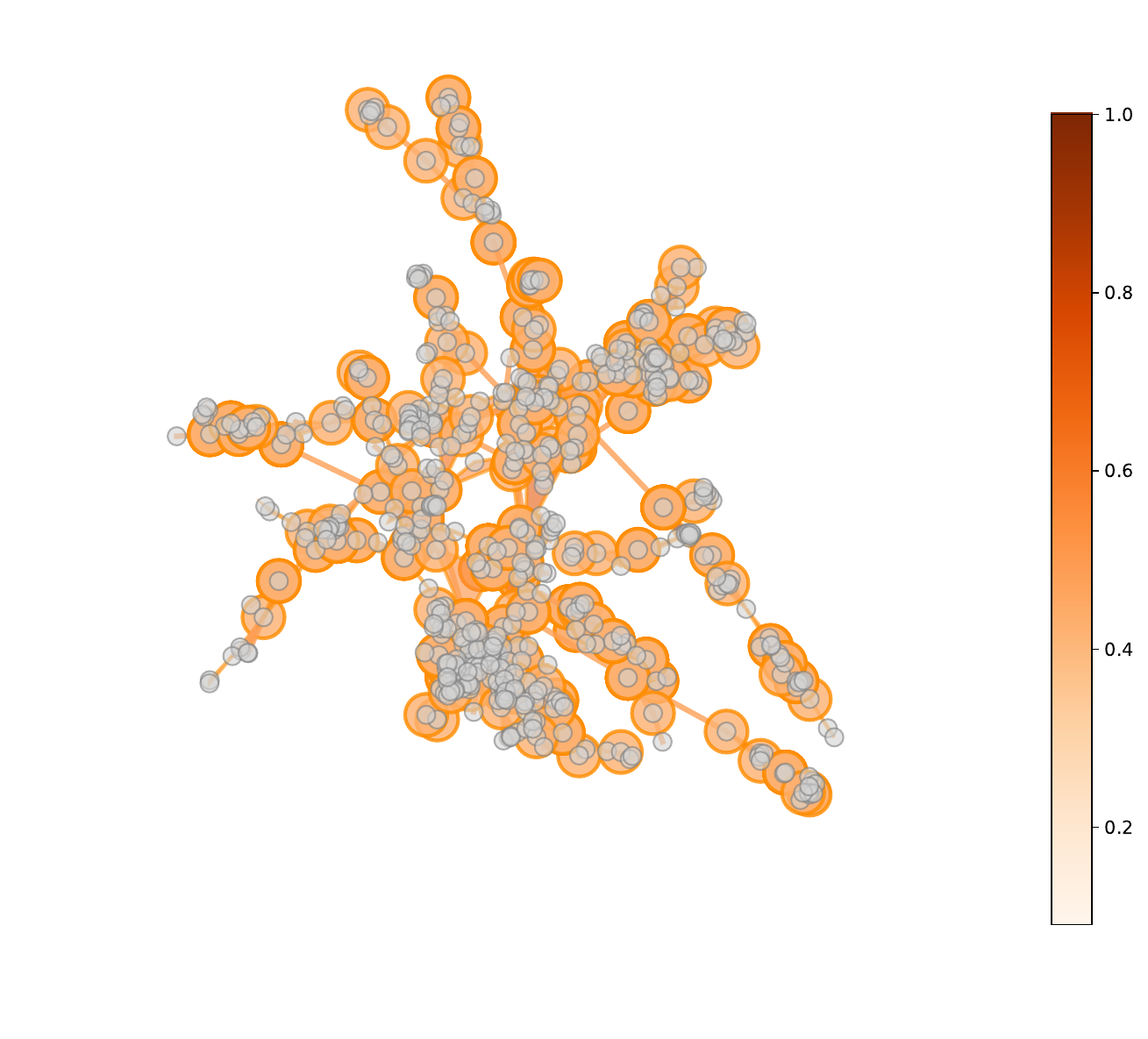} &
\raisebox{0.08\textwidth}{\rotatebox{90}{\textbf{Disease}}} \\[0.3cm]
\noalign{\hrule height 0.4pt}

\includegraphics[width=0.22\textwidth]{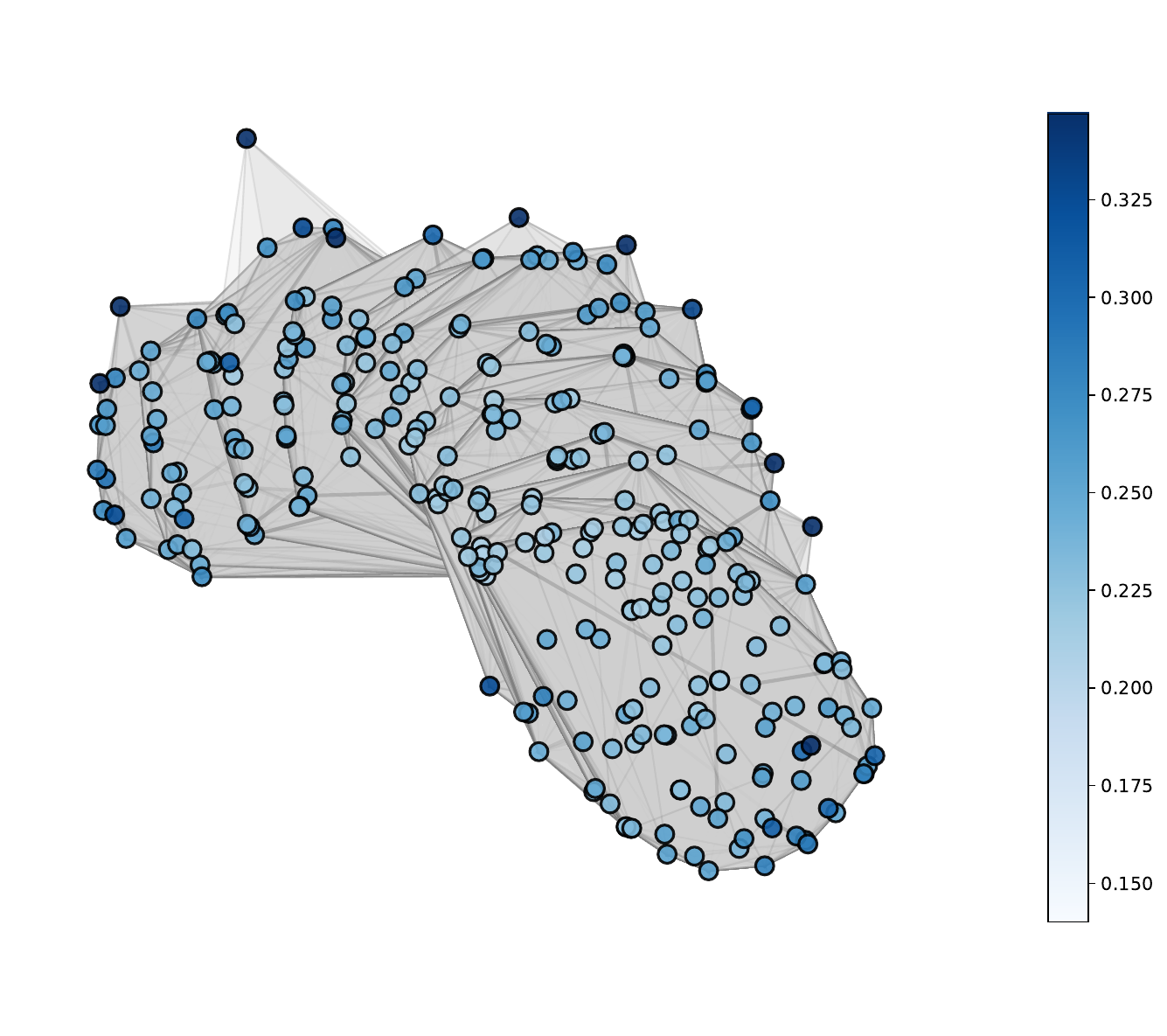} &
\includegraphics[width=0.22\textwidth]{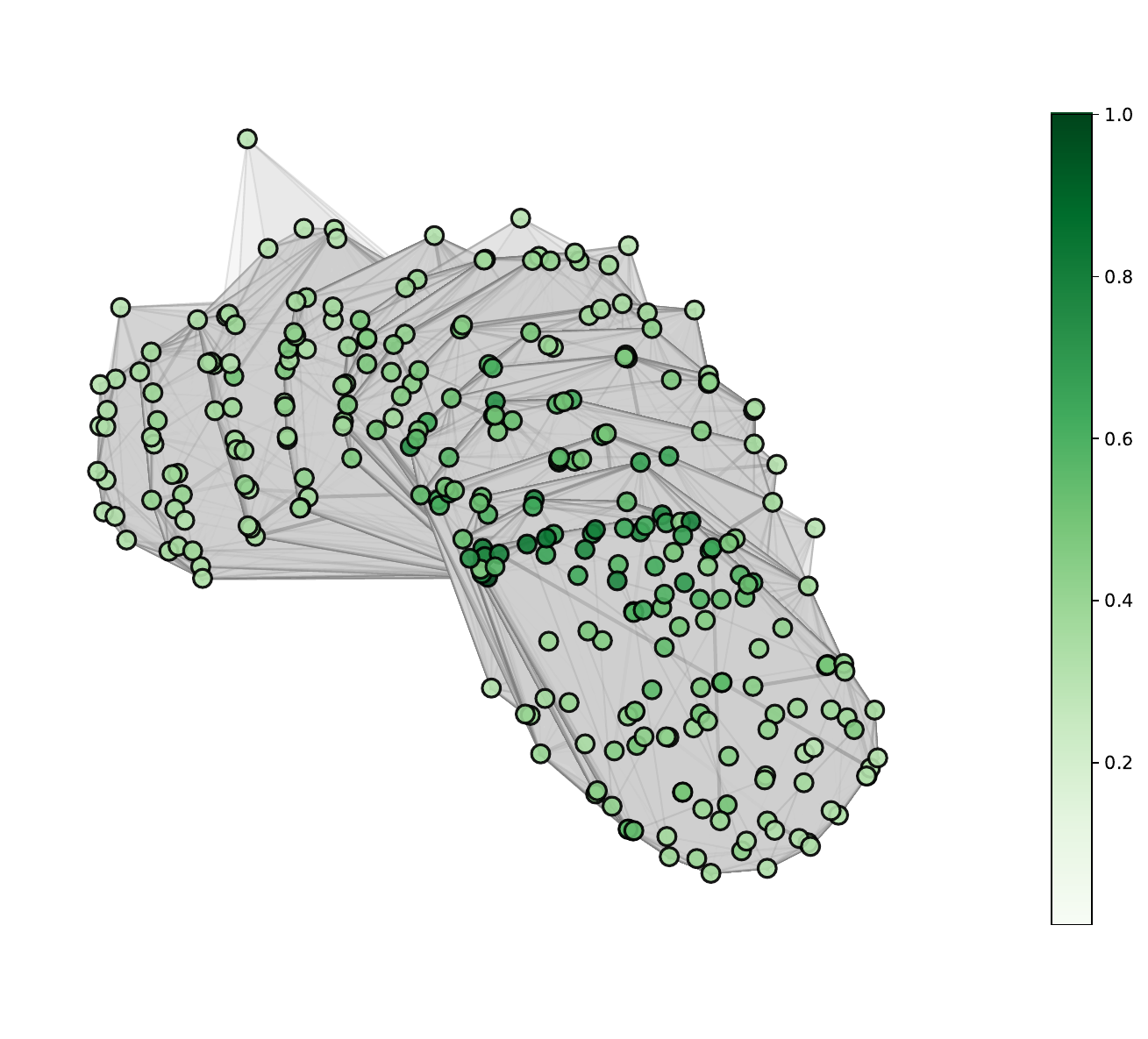} &
\includegraphics[width=0.22\textwidth]{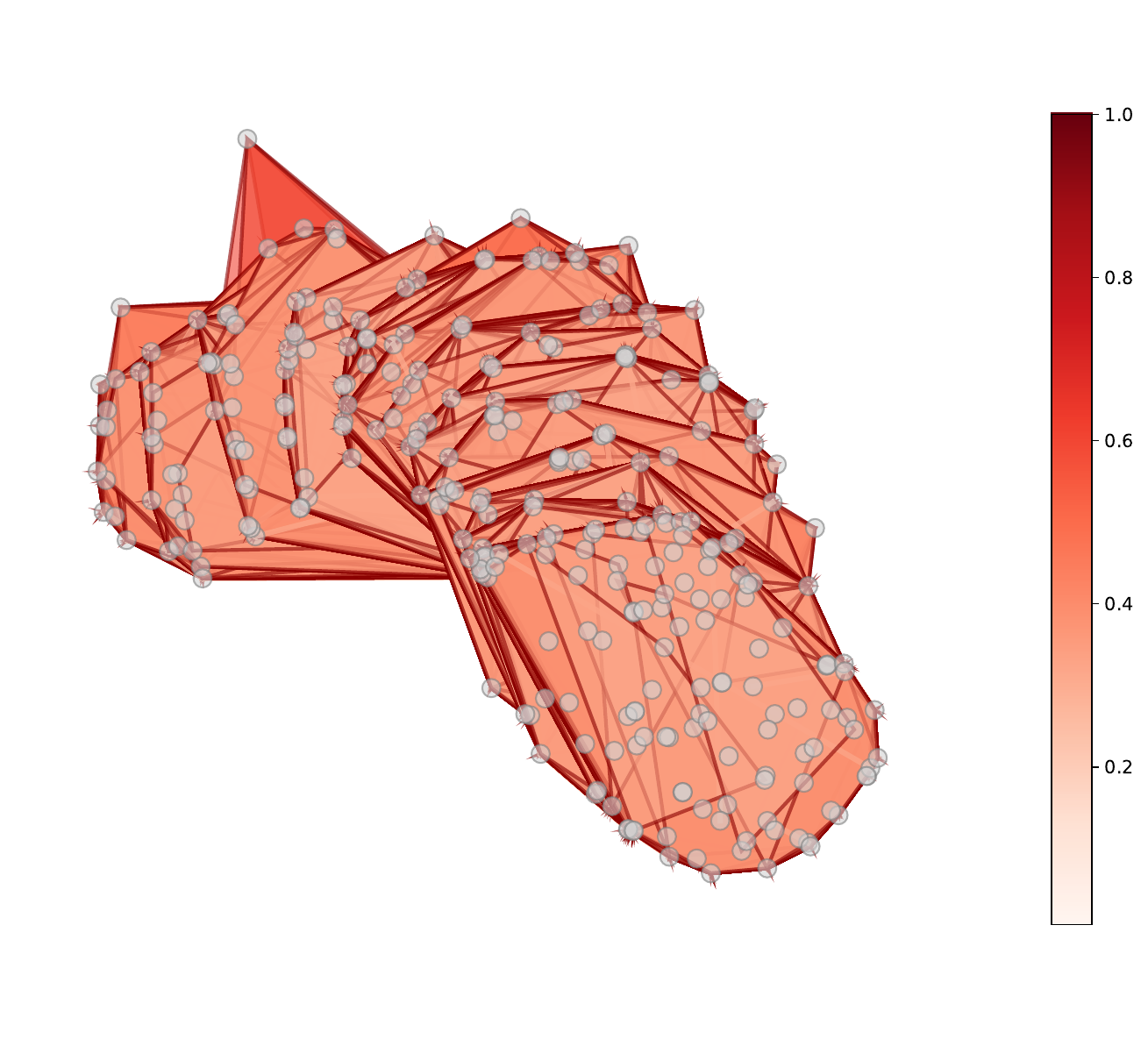} &
\includegraphics[width=0.22\textwidth]{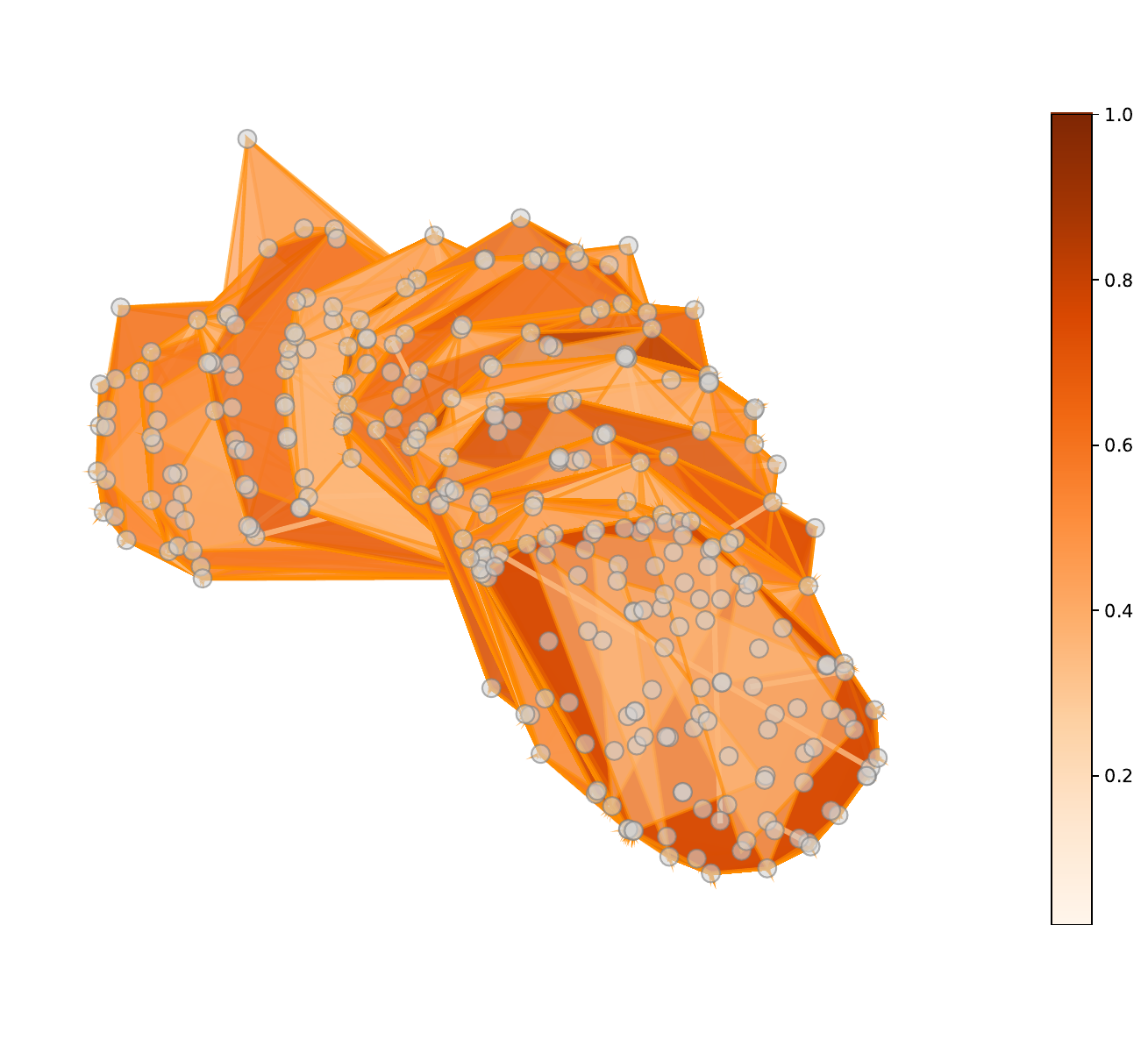} &
\raisebox{0.08\textwidth}{\rotatebox{90}{\textbf{Senate}}} \\
\bottomrule
\end{tabular}

\caption{Comparison of centrality measures across four hypergraph domains. Rows represent different networks (Pollinator, Hospital, Disease, Senate) while columns show combinations of element type (Nodes vs. Hyperedges) and centrality measure (SVD vs. Betweenness). SVD centrality produces dense, discriminative rankings across all datasets and element types, while betweenness yields sparse distributions with many zero values. Beyond density, SVD offers algorithmic determinism, its solution is mathematically unique, whereas betweenness calculations involve implementation choices for tie-breaking and path enumeration.}
\label{fig:hypergraph_comparison}
\end{figure}

The hospital contact network (second row) is dense and spatially clustered. SVD node scores are relatively compressed compared to the pollinator example: many individuals participate in similar low-frequency contact modes, resulting in a narrow—but informative—spread of SVD values that highlights a small set of relatively more central actors (for example, staff with broad contact patterns) while placing the majority near the network mean. SVD hyperedge centrality, however, shows larger variance: particular contact events or temporal aggregates serve as keystone interactions that bind otherwise local groups together and therefore obtain higher spectral weight. By contrast, hypergraph betweenness in this dataset produces a distinct pattern that emphasizes temporally or spatially constrained corridors of interaction; because hospitals feature many redundant proximity ties, betweenness tends to single out specific bridging events rather than produce a graded backbone.

The disease-association \citep{goh2007human} hypergraph (third row) is larger and more heterogeneous. SVD node centrality here reveals clear hierarchical structure: a modest core of diseases (often those sharing wide-ranging genetic or pathway overlap) project strongly onto low-frequency modes and receive elevated SVD scores, while peripheral diseases with narrow genetic overlap remain lower ranked. SVD hyperedge centrality similarly discriminates multi-disease association groups by their role in connecting the disease space. Betweenness shows complementary emphasis: it can highlight particular hyperpaths that traverse dense shared-association regions, but its spatial patterning differs from the SVD backbone, indicating that the two measures capture distinct structural notions (collective spectral participation vs.\ path-centric mediation).

The senate co-sponsorship hypergraph (fourth row) is the largest and most densely overlapping system we consider. SVD node values are relatively compressed across legislators, reflecting a political environment in which most actors participate in numerous co-sponsorship groups; nevertheless, subtle gradations in SVD score identify legislators who systematically participate in bills that span otherwise disparate co-sponsorship communities (i.e., cross-cutting legislators). SVD hyperedge centrality exhibits greater dispersion than the node field, enabling discrimination among bills by their structural role (bipartisan vs.\ highly insular co-sponsorship groups). Betweenness in this domain produces a wider numerical spread in our visualizations, typically accentuating legislators and bills that lie on many minimal co-sponsorship hyperpaths; again, however, the spatial patterning differs from the SVD field and should be read as complementary rather than contradictory.

Two observations unify these case studies. First, SVD incidence centrality provides dense, continuous rankings of both nodes and hyperedges that reveal distinct structural partitioning into hierarchical roles (generalists vs.\ specialists, backbone edges vs.\ peripheral connections) across domains. Because the SVD decomposition is obtained directly from the incidence operator and its Hodge Laplacian, the resulting centralities are deterministic and reproducible for a fixed preprocessing pipeline. Second, betweenness for hypergraphs is inherently sensitive to modelling and normalization choices: different definitions of hyperpaths or rescaling strategies change both numerical ranges and spatial emphases, which explains the domain-dependent patterns observed in the figure panels.

From a methodological perspective, these differences are expected. Spectral centralities measure how nodes and hyperedges project onto collective, low-frequency modes of the hypergraph, giving a principled account of global connectivity and redundancy. Path-based betweenness instead counts participation in minimal hyperpaths and therefore isolates mediation roles that are important for flow under shortest-path assumptions. Both viewpoints have practical value: spectral measures are particularly useful for prioritization tasks that require graded scores across many elements (e.g., conservation ranking, surveillance prioritization, or bill triage), while betweenness is useful when the domain question concerns explicit transit or mediation along minimal routes.

Across all domains, \Cref{fig:hypergraph_comparison} reveals a consistent pattern: SVD centrality produces dense, continuous rankings, while betweenness centrality yields sparse distributions with many zero-valued assignments. This density difference stems from fundamental methodological distinctions with important computational implications.

SVD centrality leverages the global spectral properties of the hypergraph incidence matrix—a well-posed mathematical problem with a unique solution (up to sign). In computational terms, SVD provides algorithmic determinism: given the same matrix, any correct implementation yields identical centrality values. This robustness is crucial for reproducible research and reliable applications.

Conversely, betweenness centrality in hypergraphs depends on shortest hyperpaths—a concept that requires algorithmic choices. Implementation decisions affect results: how are ties between equal-length paths broken? How are hyperpaths enumerated when multiple minimal edge sets connect nodes? Different software packages may make different choices, leading to inconsistent rankings. This implementation sensitivity introduces unwanted variability, especially in hypergraphs where many nodes share identical shortest-path characteristics.

From a computer science perspective, SVD's determinism and density offer practical advantages. Machine learning pipelines require stable, continuous features—sparse, implementation-dependent rankings provide poor input. Similarly, network interventions benefit from granular rankings rather than binary classifications. The spectral approach provides a mathematically robust foundation for higher-order network analysis, avoiding the implementation ambiguities of path-based measures while delivering comprehensive rankings across diverse domains.

\FloatBarrier
\section{Conclusion}

This work presented SVD incidence centrality as a unified spectral framework that addresses centrality analysis in directed networks. While undirected networks achieve theoretical equivalence with current-flow closeness, directed networks reveal the innovation of our approach: preserving directional information through spectral decomposition while maintaining connections to Hodge theory and electrical network principles.

Our key contributions span theory, methodology, and validation. Theoretically, we ground centrality measures in the Hodge Laplacian and its fundamental subspaces, connecting network analysis to established results in algebraic topology and statistical mechanics. The physical interpretation through electrical resistance analogies and energy landscapes provides new insights into how network structure influences dynamical processes.

Methodologically, we demonstrate that orientation preservation through the incidence matrix enables natural hub/authority decomposition without artificial symmetrization. The unified treatment of vertex and edge centralities within a single spectral framework ensures mathematical consistency—a critical advantage for representation learning applications.

Experimental validation across diverse network types reveals the advantages of SVD centrality over traditional betweenness measures, particularly for directed graphs where dense rankings and mathematical consistency become essential. The framework's computational efficiency through SVD algorithms and natural extension to edge centralities make it valuable for graph representation learning and network analysis applications.

The framework naturally extends to hypergraphs through higher-order incidence matrices, enabling analysis of complex systems where relationships involve multiple entities simultaneously. Our comprehensive hypergraph experiments demonstrate how the spectral approach successfully analyzes ecological pollination networks, hospital contact patterns, disease associations, and legislative co-sponsorship structures, providing meaningful centrality measures for both nodes and hyperedges while maintaining the deterministic properties that distinguish it from path-based alternatives.

Future developments include extension to temporal networks and cell, simplicial and combinatorial complexes with higher-order cells, this could help the development of topological deep learning architectures and learnable liftings \citep{franco2025differentiable}.

\section*{Acknowledgments}

This work was supported by CAPES (J. Franco, grant 88887.176396/2025-00); CNPq/Brazil (E. Tokuda, 385292/2025-2; T. Peron, 310248/2023-0; A. Dal Col, 442238/2023-1; F. Petronetto, 405903/2023-5); and FAPESP (E. Tokuda, 2025/22069-0; T. Peron, 2023/07481-6; L. Nonato, 22/09091-8).

\bibliographystyle{plainnat}
\bibliography{references}

@inproceedings{brandes2005centrality,
  title={Centrality measures based on current flow},
  author={Brandes, Ulrik and Fleischer, Daniel},
  booktitle={Annual symposium on theoretical aspects of computer science},
  pages={533--544},
  year={2005},
  organization={Springer}
}

@article{xu2021signless,
  title={Signless-laplacian eigenvector centrality: A novel vital nodes identification method for complex networks},
  author={Xu, Yan and Feng, Zhidan and Qi, Xingqin},
  journal={Pattern Recognition Letters},
  volume={148},
  pages={7--14},
  year={2021},
  publisher={Elsevier}
}

@inproceedings{kucharczuk2022pagerank,
  title={PageRank for Edges: Axiomatic Characterization},
  author={Kucharczuk, Natalia and Was, Tomasz and Skibski, Oskar},
  booktitle={Proceedings of the AAAI Conference on Artificial Intelligence},
  volume={36},
  pages={5108--5115},
  year={2022}
}

@article{zachary1977information,
  title={An information flow model for conflict and fission in small groups},
  author={Zachary, Wayne W},
  journal={Journal of Anthropological Research},
  volume={33},
  number={4},
  pages={452--473},
  year={1977},
  publisher={University of New Mexico},
  doi={10.1086/jar.33.4.3629752}
}

@inproceedings{
franco2025differentiable,
title={Differentiable Lifting for Topological Neural Networks},
author={Jorge Luiz Franco and Gabriel Duarte and Alexander V Nikitin and Moacir A Ponti and Diego Mesquita and Amauri H Souza},
booktitle={Non-Euclidean Foundation Models: Advancing AI Beyond Euclidean Frameworks},
year={2025},
url={https://openreview.net/forum?id=qL6sgVeOaI}
}

@article{franco2026holding,
  author  = {Franco, Jorge L. and Machado Neto, Manoel V. and Verri, Filipe A. N. and Amancio, Diego R.},
  title   = {Graph machine learning for flight delay prediction due to holding manouver},
  journal = {Physica A: Statistical Mechanics and its Applications},
  pages   = {131318},
  year    = {2026},
  issn    = {0378-4371},
  doi     = {10.1016/j.physa.2026.131318},
  url     = {https://www.sciencedirect.com/science/article/pii/S0378437126000543}
}

@article{bockholt2021systematic,
  author  = {Bockholt, Mareike and Zweig, Katharina A.},
  title   = {A systematic evaluation of assumptions in centrality measures by empirical flow data},
  journal = {Social Network Analysis and Mining},
  volume  = {11},
  pages   = {25},
  year    = {2021},
  doi     = {10.1007/s13278-021-00725-3}
}

@article{fanuel2019deformed,
  author    = {Fanuel, Micha{\"e}l and Suykens, Johan A. K.},
  title     = {Deformed Laplacians and Spectral Ranking in Directed Networks},
  journal   = {Applied and Computational Harmonic Analysis},
  volume    = {47},
  number    = {2},
  pages     = {397--422},
  year      = {2019},
  doi       = {10.1016/j.acha.2017.09.002}
}

@article{Gao2014,
  title     = {Target control of complex networks},
  author    = {Gao, Jianxi and Liu, Yang-Yu and D'Souza, Raissa M. and Barab{\'a}si, Albert-L{\'a}szl{\'o}},
  journal   = {Nature Communications},
  volume    = {5},
  number    = {1},
  pages     = {5415},
  year      = {2014},
  publisher = {Nature Publishing Group},
  doi       = {10.1038/ncomms6415},
  url       = {https://www.nature.com/articles/ncomms6415}
}

@article{schaub2020random,
  author    = {Schaub, Michael T. and Benson, Austin R. and Horn, Paul and Lippner, Gabor and Jadbabaie, Ali},
  title     = {Random Walks on Simplicial Complexes and the Normalized Hodge 1-Laplacian},
  journal   = {SIAM Review},
  volume    = {62},
  number    = {2},
  pages     = {353--391},
  year      = {2020},
  doi       = {10.1137/18M1201019}
}

@article{newman2004finding,
  title={Finding and evaluating community structure in networks},
  author = {Newman, M E J and Girvan, M},
  journal={Physical review E},
  volume={69},
  number={2},
  pages={026113},
  year={2004},
  publisher={APS}
}

@article{vasilyeva2024matrix,
  title={Matrix centrality for annotated hypergraphs},
  author={Vasilyeva, E and Samoylenko, I and Kovalenko, K and Musatov, D and Raigorodskii, AM and Boccaletti, S},
  journal={Chaos, Solitons \& Fractals},
  volume={186},
  pages={115256},
  year={2024},
  publisher={Elsevier}
}

@article{snijders2010introduction,
  title={Introduction to stochastic actor-based models for network dynamics},
  author={Snijders, Tom A B and van de Bunt, Gerhard G and Steglich, Christian E G},
  journal={Social Networks},
  volume={32},
  number={1},
  pages={44--60},
  year={2010},
  publisher={Elsevier},
  doi={10.1016/j.socnet.2009.02.004}
}

@article{friedman1998computing,
  title={Computing Betti numbers via combinatorial Laplacians},
  author={Friedman, Joel},
  journal={Algorithmica},
  volume={21},
  number={4},
  pages={331--346},
  year={1998}
}

@misc{openflights2024,
  title={OpenFlights: Flight logging, mapping, stats and sharing},
  author={{OpenFlights}},
  year={2024},
  url={https://openflights.org/data.html},
  note={Accessed 2024}
}

@article{subelj2011robust,
  title={Robust network community detection using balanced propagation},
  author={{\v{S}}ubelj, Lovro and Bajec, Marko},
  journal={The European Physical Journal B},
  volume={81},
  number={3},
  pages={353--362},
  year={2011},
  publisher={Springer},
  doi={10.1140/epjb/e2011-10979-2}
}

@article{coulomb2005gene,
  title={Gene essentiality and the topology of protein interaction networks},
  author={Coulomb, S and Bauer, M and Bernard, D and Marsolier-Kergoat, M-C},
  journal={Proceedings of the Royal Society B: Biological Sciences},
  volume={272},
  number={1573},
  pages={1721--1725},
  year={2005},
  publisher={The Royal Society},
  doi={10.1098/rspb.2005.3128}
}

@article{duch2005community,
  title={Community identification using Extremal Optimization},
  author={Duch, Jordi and Arenas, Alex},
  journal={Physical Review E},
  volume={72},
  number={2},
  pages={027104},
  year={2005},
  publisher={APS},
  doi={10.1103/physreve.72.027104}
}

@article{contreras2024beyond,
  title={Beyond directed hypergraphs: heterogeneous hypergraphs and spectral centralities},
  author={Contreras-Aso, Gonzalo and Criado, Regino and Romance, Miguel},
  journal={Journal of Complex Networks},
  volume={12},
  number={4},
  pages={cnae037},
  year={2024},
  publisher={Oxford University Press}
}

@article{brohl2022straightforward,
  title={A straightforward edge centrality concept derived from generalizing degree and strength},
  author={Br{\"o}hl, Timo and Lehnertz, Klaus},
  journal={Scientific Reports},
  volume={12},
  number={1},
  pages={4407},
  year={2022},
  publisher={Nature Publishing Group UK London}
}

@article{dal2023graph,
  title={Graph regularization centrality},
  author={Dal Col, Alcebiades and Petronetto, Fabiano},
  journal={Physica A: Statistical Mechanics and its Applications},
  volume={628},
  pages={129188},
  year={2023},
  publisher={Elsevier}
}

@article{eigenHypergraph,
author = {Benson, Austin R.},
title = {Three Hypergraph Eigenvector Centralities},
journal = {SIAM Journal on Mathematics of Data Science},
volume = {1},
number = {2},
pages = {293-312},
year = {2019},
doi = {10.1137/18M1203031},

URL = {

        https://doi.org/10.1137/18M1203031



},
eprint = {

        https://doi.org/10.1137/18M1203031



}
}

@mastersthesis{Bek2006PollinationNetwork,
  author  = {Bek, S.},
  title   = {A Pollination Network from a Danish Forest Meadow},
  school  = {University of Aarhus},
  address = {Aarhus, Denmark},
  year    = {2006}
}

@article{vanhems2013estimating,
  title={Estimating potential infection transmission routes in hospital wards using wearable proximity sensors},
  author={Vanhems, Philippe and Barrat, Alain and Cattuto, Ciro and Pinton, Jean-Fran{\c{c}}ois and Khanafer, Nagham and R{\'e}gis, Corinne and Kim, Byeul-a and Comte, Brigitte and Voirin, Nicolas},
  journal={PloS one},
  volume={8},
  number={9},
  pages={e73970},
  year={2013},
  publisher={Public Library of Science San Francisco, USA}
}

@article{goh2007human,
author = {Kwang-Il Goh  and Michael E. Cusick  and David Valle  and Barton Childs  and Marc Vidal  and Albert-László Barabási },
title = {The human disease network},
journal = {Proceedings of the National Academy of Sciences},
volume = {104},
number = {21},
pages = {8685-8690},
year = {2007},
doi = {10.1073/pnas.0701361104},
URL = {https://www.pnas.org/doi/abs/10.1073/pnas.0701361104},
eprint = {https://www.pnas.org/doi/pdf/10.1073/pnas.0701361104}
}

@misc{landry2024senate-bills,
  author       = {Landry, N.},
  title        = {{senate-bills} (v0.0) [Data set]},
  year         = {2024},
  publisher    = {Zenodo},
  doi          = {10.5281/zenodo.10957697},
  url          = {https://doi.org/10.5281/zenodo.10957697},
  note         = {Version 0.0}
}

@inproceedings{hyperbetw,
author = {Lee, Kwang Hee and Kim, Myoung Ho},
title = {Computing Betweenness Centrality in B-hypergraphs},
year = {2017},
isbn = {9781450349185},
publisher = {Association for Computing Machinery},
address = {New York, NY, USA},
url = {https://doi.org/10.1145/3132847.3133093},
doi = {10.1145/3132847.3133093},
booktitle = {Proceedings of the 2017 ACM on Conference on Information and Knowledge Management},
pages = {2147–2150},
numpages = {4},
location = {Singapore, Singapore},
series = {CIKM '17}
}

@book{biggs1993algebraic,
  title={Algebraic Graph Theory},
  author={Biggs, Norman},
  publisher={Cambridge University Press},
  year={1993}
}

@article{lim2019hodgelaplaciansgraphs,
  title={Hodge laplacians on graphs},
  author={Lim, Lek-Heng},
  journal={SIAM Review},
  volume={62},
  number={4},
  pages={685--715},
  year={2020},
  publisher={SIAM}
}

@misc{alain2024graphclassificationgaussianprocesses,
      title={Graph Classification Gaussian Processes via Hodgelet Spectral Features},
      author={Mathieu Alain and So Takao and Xiaowen Dong and Bastian Rieck and Emmanuel Noutahi},
      year={2024},
      eprint={2410.10546},
      archivePrefix={arXiv},
      primaryClass={cs.LG},
      url={https://arxiv.org/abs/2410.10546},
}

@article{brin1998pagerank,
  title={The PageRank citation ranking: bringing order to the web},
  author={Brin, Sergey},
  journal={Proceedings of ASIS, 1998},
  volume={98},
  pages={161--172},
  year={1998}
}

@article{kleinberg1999authoritative,
  title={Authoritative sources in a hyperlinked environment},
  author={Kleinberg, Jon M},
  journal={Journal of the ACM (JACM)},
  volume={46},
  number={5},
  pages={604--632},
  year={1999},
  publisher={ACM New York, NY, USA}
}

@article{tudisco2021hyperedge,
  title={Node and edge nonlinear eigenvector centrality for hypergraphs},
  author={Tudisco, Francesco and Higham, Desmond J},
  journal={Communications Physics},
  volume={4},
  number={1},
  pages={201},
  year={2021},
  publisher={Nature Publishing Group UK London}
}

@Inbook{Lu2013edgebetw,
author="Lu, LongJason
and Zhang, Minlu",
title="Edge Betweenness Centrality",
bookTitle="Encyclopedia of Systems Biology",
year="2013",
publisher="Springer New York",
address="New York, NY",
pages="647--648",
isbn="978-1-4419-9863-7",
doi="10.1007/978-1-4419-9863-7_874",
url="https://doi.org/10.1007/978-1-4419-9863-7_874"
}

@article{bonacich1987power,
  title={Power and centrality: A family of measures},
  author={Bonacich, Phillip},
  journal={American journal of sociology},
  volume={92},
  number={5},
  pages={1170--1182},
  year={1987},
  publisher={University of Chicago Press}
}

@article{freeman1978centrality,
title = {Centrality in social networks conceptual clarification},
journal = {Social Networks},
volume = {1},
number = {3},
pages = {215-239},
year = {1978},
issn = {0378-8733},
doi = {https://doi.org/10.1016/0378-8733(78)90021-7},
url = {https://www.sciencedirect.com/science/article/pii/0378873378900217},
author = {Linton C. Freeman}
}

@article{liu2023eigenvector,
  title={Eigenvector centrality in simplicial complexes of hypergraphs},
  author={Liu, Xiaolu and Zhao, Chong},
  journal={Chaos: An Interdisciplinary Journal of Nonlinear Science},
  volume={33},
  number={9},
  year={2023},
  publisher={AIP Publishing}
}

@article{baglama2014analysis,
  title={Analysis of directed networks via partial singular value decomposition and Gauss quadrature},
  author={Baglama, J and Fenu, Caterina and Reichel, L and Rodriguez, Giuseppe},
  journal={Linear Algebra and its Applications},
  volume={456},
  pages={93--121},
  year={2014},
  publisher={Elsevier}
}

@article{benzi2013ranking,
  title={Ranking hubs and authorities using matrix functions},
  author={Benzi, Michele and Estrada, Ernesto and Klymko, Christine},
  journal={Linear Algebra and its Applications},
  volume={438},
  number={5},
  pages={2447--2474},
  year={2013},
  publisher={Elsevier}
}

@article{sarkar2011community,
  title={Community detection in graphs using singular value decomposition},
  author={Sarkar, Somwrita and Dong, Andy},
  journal={Physical Review E—Statistical, Nonlinear, and Soft Matter Physics},
  volume={83},
  number={4},
  pages={046114},
  year={2011},
  publisher={APS}
}

@article{klein1993resistance,
  title={Resistance distance},
  author={Klein, Douglas J and Randi{\'c}, Milan},
  journal={Journal of mathematical chemistry},
  volume={12},
  number={1},
  pages={81--95},
  year={1993},
  publisher={Springer}
}

@book{chung1997spectral,
  title={Spectral graph theory},
  author={Chung, Fan RK},
  volume={92},
  year={1997},
  publisher={American Mathematical Society}
}

@article{horak2013spectra,
  title={Spectra of combinatorial Laplace operators on simplicial complexes},
  author={Horak, Danijela and Jost, J{\"u}rgen},
  journal={Advances in Mathematics},
  volume={244},
  pages={303--336},
  year={2013},
  publisher={Elsevier}
}

@article{eckmann1945harmonische,
  title={{\"U}ber die harmonischen Beziehungen zwischen den Randoperatoren und den Hauptoperatoren der Hodge-de-Rham-Theorie},
  author={Eckmann, Beno},
  journal={Commentarii Mathematici Helvetici},
  volume={19},
  pages={271--288},
  year={1945},
  publisher={Springer}
}

\appendix

\section{Proof of Effective Resistance Identity}
\label{app:proof}

\begin{theorem}[Effective Resistance and SVD Centrality]
\label{thm:effective_resistance}
Let $G$ be a connected undirected graph with $n$ vertices and incidence matrix $B$ having compact SVD $B = U\Sigma V^\top$. For any vertex $i$, the sum of effective resistances from $i$ to all other vertices satisfies
\[
\sum_{j=1}^n R_{ij} = n C_v(i) + \operatorname{tr}(L_0^+),
\]
where $C_v(i) = [L_0^+]_{ii}$ is the SVD vertex centrality and $L_0^+ = U\Sigma^{-2}U^\top$ is the pseudoinverse of the vertex Hodge Laplacian.
\end{theorem}

\begin{proof}
From the definition of effective resistance, we have:
\[
R_{ij} = (e_i - e_j)^\top L_0^+ (e_i - e_j).
\]
Substituting the SVD-based expression for $L_0^+ = U \Sigma^{-2} U^\top$ yields:
\[
R_{ij} = (e_i - e_j)^\top U \Sigma^{-2} U^\top (e_i - e_j).
\]
Define $w = U^\top (e_i - e_j)$, so that
\[
R_{ij} = w^\top \Sigma^{-2} w = \sum_{k=1}^r \frac{w_k^2}{\sigma_k^2},
\]
where $w_k = u_k^\top (e_i - e_j) = u_{k,i} - u_{k,j}$. Thus,
\[
R_{ij} = \sum_{k=1}^r \frac{(u_{k,i} - u_{k,j})^2}{\sigma_k^2}.
\]

Now, consider the sum of effective resistances from node $i$ to all other nodes:
\[
\sum_{j=1}^n R_{ij} = \sum_{j=1}^n \sum_{k=1}^r \frac{(u_{k,i} - u_{k,j})^2}{\sigma_k^2}.
\]
Interchanging the order of summation:
\[
\sum_{j=1}^n R_{ij} = \sum_{k=1}^r \frac{1}{\sigma_k^2} \sum_{j=1}^n (u_{k,i} - u_{k,j})^2.
\]
Expanding the inner sum:
\[
\sum_{j=1}^n (u_{k,i} - u_{k,j})^2 = \sum_{j=1}^n (u_{k,i}^2 - 2 u_{k,i} u_{k,j} + u_{k,j}^2) = n u_{k,i}^2 - 2 u_{k,i} \sum_{j=1}^n u_{k,j} + \sum_{j=1}^n u_{k,j}^2.
\]

Since the columns of $U$ are orthonormal, $\sum_{j=1}^n u_{k,j}^2 = 1$. Additionally, for the graph Laplacian's pseudoinverse, the all-ones vector $\mathbf{1}$ is in the null space of $L_0$, and hence $U^\top \mathbf{1} = 0$, implying $\sum_{j=1}^n u_{k,j} = 0$ for all $k$ corresponding to non-zero singular values. Therefore,
\[
\sum_{j=1}^n (u_{k,i} - u_{k,j})^2 = n u_{k,i}^2 + 1.
\]

Substituting back:
\[
\sum_{j=1}^n R_{ij} = \sum_{k=1}^r \frac{1}{\sigma_k^2} (n u_{k,i}^2 + 1) = n \sum_{k=1}^r \frac{u_{k,i}^2}{\sigma_k^2} + \sum_{k=1}^r \frac{1}{\sigma_k^2}.
\]

Recognizing that $\sum_{k=1}^r \frac{u_{k,i}^2}{\sigma_k^2} = [L_0^+]_{ii} = C_v(i)$ and $\sum_{k=1}^r \frac{1}{\sigma_k^2} = \operatorname{tr}(L_0^+)$, we obtain:
\[
\sum_{j=1}^n R_{ij} = n C_v(i) + \operatorname{tr}(L_0^+),
\]
which completes the proof.
\end{proof}

\section{Experimental Methodology and Implementation Details}
\label{app:methodology}

This appendix provides comprehensive details on our experimental methodology, implementation choices, and detailed pseudocode \Cref{alg:svd_centrality}.

To ensure reproducibility, the complete implementation of our centrality framework and experiments is publicly available at: \texttt{https://github.com/JorgeLuizFranco/svd-centrality}.

All experiments were conducted using Python 3.8+ with key libraries including NetworkX 3.1 for graph construction and traditional centrality computations, NumPy 1.21+ for numerical computations and matrix operations, SciPy 1.7+ for SVD decomposition and pseudoinverse computation, Matplotlib 3.5+ for publication-quality figures, and Seaborn 0.11+ for statistical visualization.

The core SVD centrality computation follows Algorithm~\ref{alg:svd_centrality}, where we construct the oriented incidence matrix, compute the compact SVD, derive the pseudoinverse Laplacians, and extract the diagonal elements as centrality measures. Hub and authority centralities are then computed by aggregating edge centralities to vertices according to the incoming and outgoing edge sets.

\begin{algorithm}
\caption{SVD Incidence Centrality Computation}
\label{alg:svd_centrality}
\begin{algorithmic}[1]
\REQUIRE Directed graph $G = (V, E)$ with $n$ vertices and $m$ edges
\ENSURE Vertex centralities $C_v$, edge centralities $C_e$, hub scores $H$, authority scores $A$
\STATE Construct oriented incidence matrix $B \in \mathbb{R}^{n \times m}$
\STATE Compute compact SVD: $B = U \Sigma V^T$ where $U \in \mathbb{R}^{n \times r}$, $V \in \mathbb{R}^{m \times r}$
\STATE Apply Tikhonov regularization: $\sigma_{k,\text{reg}}^{-2} = \frac{1}{\sigma_k^2 + \tau}$ for $k = 1, \ldots, r$
\STATE Compute vertex centrality: $C_v(i) = \sum_{k=1}^r \frac{u_{k,i}^2}{\sigma_k^2 + \tau}$ for $i = 1, \ldots, n$
\STATE Compute edge centrality: $C_e(j) = \sum_{k=1}^r \frac{v_{k,j}^2}{\sigma_k^2 + \tau}$ for $j = 1, \ldots, m$
\FOR{each vertex $i = 1, \ldots, n$}
    \STATE $H(i) = \sum_{j \in E_{\text{out}}(i)} C_e(j)$ \COMMENT{Hub score: outgoing edge centralities}
    \STATE $A(i) = \sum_{j \in E_{\text{in}}(i)} C_e(j)$ \COMMENT{Authority score: incoming edge centralities}
\ENDFOR
\STATE Apply convex combination: $H_{\alpha}(i) = \alpha C_v(i) + (1-\alpha) H(i)$
\STATE Apply convex combination: $A_{\alpha}(i) = \alpha C_v(i) + (1-\alpha) A(i)$
\RETURN $C_v, C_e, H_{\alpha}, A_{\alpha}$
\end{algorithmic}
\end{algorithm}




\end{document}